\documentclass[draftclsnofoot, journal, letterpaper, onecolumn, 12pt]{IEEEtran}
\usepackage{ifpdf}
\usepackage{cite}

\usepackage{color}
\usepackage{amsthm}
\usepackage{amsmath,amssymb}
\usepackage{algorithmic}
\usepackage{array}
\usepackage{mdwmath}
\usepackage{mdwtab}
\usepackage{graphicx}
\usepackage{lipsum}
\usepackage{epstopdf }
\usepackage{algorithm,algorithmic}
\usepackage{multicol}
\usepackage{multirow}
\usepackage{subfigure}

\usepackage{enumerate}
\usepackage{eqparbox}
\usepackage{url}

\usepackage{subfig}
\usepackage{amsbsy}
\usepackage{setspace}
\usepackage{comment}

\usepackage[font=small, figurename=Fig.]{caption}

\usepackage{multirow}
\usepackage{xcolor}
\usepackage{amsmath,graphicx}
\usepackage{amssymb}
\usepackage{subfigure}
\usepackage{footmisc}
\usepackage{lipsum}
\usepackage{mathtools}
\usepackage{cuted}
\usepackage{esint}
\usepackage{cite}
\usepackage{etoolbox}

\let\mybibitem\bibitem
\renewcommand{\bibitem}[1]{%
\ifstrequal{#1}{ORAN1}{\color{black}\mybibitem{#1}}
{\ifstrequal{#1}{tun2021energy}{\color{black}\mybibitem{#1}}
{\ifstrequal{#1}{yang2020multiUAV}{\color{black}\mybibitem{#1}}
{\ifstrequal{#1}{zhang2020computation}{\color{black}\mybibitem{#1}}
{\ifstrequal{#1}{cooper2017comparative}{\color{black}\mybibitem{#1}}
{\ifstrequal{#1}{IBM}{\color{black}\mybibitem{#1}}
{\ifstrequal{#1}{jovsilo2022joint}{\color{black}\mybibitem{#1}}
{\ifstrequal{#1}{guo2022distributed}{\color{black}\mybibitem{#1}}
{\ifstrequal{#1}{wong2015improving}{\color{black}\mybibitem{#1}}
{\ifstrequal{#1}{Yigitoglu2017foggy}{\color{black}\mybibitem{#1}}
{\ifstrequal{#1}{yang2019multi}{\color{black}\mybibitem{#1}}
{\color{black}\mybibitem{#1}}}%
}}}}}}}}}}

\def\cI{\mathcal{I}}
\def\cJ{\mathcal{J}}

\def\cO{\mathcal{O}}

\def\argmax{\operatornamewithlimits{arg\,max}}

\newtheorem{theorem}{Theorem}%
\newtheorem{lemma}{Lemma}%

\begin{document}

\title{An Online Framework for Ephemeral Edge Computing in the Internet of Things  \vspace{-0mm}}

\author{Gilsoo~Lee,~\IEEEmembership{Member,~IEEE}, Walid~Saad,~\IEEEmembership{Fellow,~IEEE}, Mehdi~Bennis,~\IEEEmembership{Fellow,~IEEE},
Cheonyong~Kim, and Minchae~Jung,~\IEEEmembership{Member,~IEEE}% <-this % stops a space

\thanks{ 
{\hspace{-0.52cm} A preliminary version of this  paper was presented in \cite{lee2018online}.}

{\textcolor{black}{This research was supported by the U.S. National Science Foundation under Grant CNS-1814477
and by the Basic Science Research Program through the National Research Foundation of Korea (NRF) funded by the Ministry of Education (NRF-2021R1C1C1012950).}}

{G. Lee was with Department of Electrical and Computer Engineering, Virginia Tech, Blacksburg, VA 24061 USA (e-mail: gilsoolee@vt.edu).}

{W. Saad is with Wireless@VT, Department of Electrical and Computer Engineering, Virginia Tech, Blacksburg, VA 24061 USA (e-mail: walids@vt.edu).}

{M. Bennis is with Centre for Wireless Communications, University of Oulu, Finland (email: mehdi.bennis@oulu.fi).}

{C. Kim and M. Jung are with the Department of Electronics and Information Engineering, Sejong
University, Seoul (e-mail: \{cykim0807, mcjung\}@sejong.ac.kr).}

\vspace{-0.05cm}
}
}
\markboth{}%'
{Shell \MakeLowercase{\textit{et al.}}: Bare Demo of IEEEtran.cls for Journals}
%\author{\IEEEauthorblockN{Gilsoo~Lee$^{*}$,~Walid~Saad$^{*}$,~Mehdi~Bennis$^{\dag}$,~Cheonyong~Kim$^{\ddagger}$,~Minchae~Jung$^{\ddagger}$}\\
%\IEEEauthorblockA{  $^{*}$ Wireless@VT, Department of Electrical and Computer Engineering, Virginia Tech, Blacksburg, VA, USA, \\
%Emails: \protect\url{{gilsoolee, walids}@vt.edu}. \\
% $^\dag$ Centre for Wireless Communications, University of Oulu, Finland, \\Email: \url{mehdi.bennis@oulu.fi}.\\
%  $^\ddagger$ Department of Electronic Engineering, Soonchunhyang University, South Korea, \\Email: \protect\url{{cykim0807, hosaly}@sch.ac.kr}.\\
%}

\maketitle
\thispagestyle{empty}
\vspace{-5mm}
\begin{abstract} 
In the Internet of Things (IoT) environment, edge computing can be initiated at anytime and anywhere.
However, in an IoT \textcolor{black}{environment}, edge computing sessions are often \emph{ephemeral}, i.e., they last for a short period of time and can often be discontinued once the current application usage is completed or the edge devices leave the system due to factors such as mobility. 
Therefore, in this paper, the  problem of \emph{ephemeral edge computing} in an IoT is studied by considering scenarios in which edge computing operates within a limited time period. 
To this end, a novel online framework is proposed in which a source edge node offloads its computing tasks from sensors within an area  to neighboring edge nodes for distributed task computing, within the limited period of time of an ephemeral edge computing system. 
The online nature of the framework allows the edge nodes to optimize their task allocation and decide on which neighbors to use for task processing, even when the tasks are revealed to the source edge node in an online manner, and the information on future task arrivals is unknown. 
The proposed framework essentially maximizes the number of computed tasks by jointly considering the communication and computation latency. 
To solve the \textcolor{black}{joint optimization}, an online greedy algorithm is proposed and solved by using the primal-dual approach. 
Since the primal problem provides an upper bound of the original dual problem, the competitive ratio of the online approach is analytically derived as a function of the task sizes and the data rates of the edge nodes. 
Simulation results show that the proposed online algorithm can achieve a near-optimal task allocation with an optimality gap that is no higher than $7.1$\% compared to the offline, optimal solution with complete knowledge of all tasks. 
\end{abstract}

\vspace{-0.15cm}
\begin{IEEEkeywords}
\vspace{-0.2cm}
Competitive ratio, edge computing, internet of things (IoT), online optimization, task allocation.
\end{IEEEkeywords}

\IEEEpeerreviewmaketitle
\section{Introduction}

Next-generation wireless networks will bring in new Internet of Things (IoT) services that can potentially transform people's daily lives \cite{saad2019vision,karimzadeh2022common}.
Much of these emerging IoT and 5G \textcolor{black}{(fifth generation of wireless communications)} services require low latency in terms of both communication and computing. 
To deliver low-latency IoT services, one can resort to edge computing \cite{chiang2016fog, chen2019joint} techniques that can use radio and computing resources at \textcolor{black}{a network edge}\footnote{\textcolor{black}{According to the network environment and application scenario, the network’s edge can include various entities such as border routers, access points, base stations, mobile devices, and connected vehicles. In this study, we focused on an edge network consisting of mobile nodes.}}.

In particular, by using local computing resources, edge computing can significantly reduce the distance of data transmission, thus inducing smaller communication latency.
To enable large-scale and distributed edge computing among heterogeneous devices, there is a need to enable edge devices to pool their computing resources by instantaneously forming a local edge network to process the computational tasks received from various user applications \cite{chiang2016fog}. 
Clearly, if properly deployed, edge computing will bring forth key benefits for low-latency IoT services by ensuring that a local edge network is instantaneously deployed by edge devices. 
Therein, fundamental challenges include joint radio and computing resource management and application-oriented edge computing system and architecture design. 

\vspace{-5mm}
\subsection{\textcolor{black}{Related Work}}

\subsubsection{Edge computing in general IoT environments}
Edge computing enables a diverse set of IoT services ranging from real-time IoT applications running on user devices to safety applications operating on connected vehicles \cite{hu2015mobile}.
\textcolor{black}{Recently, a number of edge computing proof of concepts have been implemented for various IoT applications such as network resource management \cite{wong2015improving}, IoT application deployment \cite{Yigitoglu2017foggy}, and multimedia data caching \cite{yang2019multi}.
The work in \cite{IBM} showed how one can deploy, in the real world, edge devices with powerful computing resources and an inherent capability of running computation intensive applications.}
\textcolor{black}{Recent prior works} in \cite{wang2019qos, huang2019deep, huang2017fair, mao2016dynamic, kuang2019partial, wang2017computation, chen2018virtual} studied deployment scenarios and resource allocation problems for standard edge computing in static or low-mobility networks. 
In particular, the work in \cite{wang2019qos} proposed an edge computing platform deployed in network infrastructure nodes such as base stations to provide contents to users while maintaining a required quality-of-service. 
Meanwhile, the authors in  \cite{huang2019deep} studied the problem of joint computational task offloading and radio resource allocation in a wireless powered edge computing system by using deep learning. 
The work in \cite{huang2017fair} introduced a caching scheme so as to maximize fairness for an edge computing environment consisting of heterogeneous devices with different communication and computing resources.  
The authors in \cite{mao2016dynamic} proposed a Lyapunov optimization-based computation offloading algorithm to jointly control transmit power and CPU \textcolor{black}{(Central Processing Unit)}-clock speeds when edge computing devices are powered by energy harvesting techniques. 
The work in \cite{kuang2019partial} studied a partial computational task offloading and radio allocation problems are jointly studied. 
Moreover, in \cite{wang2017computation}, a joint strategy of computational offloading and content caching is proposed to maximize the utilization of each \textcolor{black}{edge node radio} and computing resources when the statistical information on the content request is previously known. 
In \cite{chen2018virtual}, the authors used edge computing for enhancing virtual reality services.
%\textcolor{blue}{Recently, the joint problem of task assignment and radio resource allocation has been studied in \cite{jovsilo2022joint} and \cite{guo2022distributed}. In particular, the work in \cite{jovsilo2022joint} studies the management of radio and computing resources in network slicing. Also, the authors in \cite{guo2022distributed} study task offloading and radio resource allocation to securely offload the tasks having different priorities in cognitive eavesdropping environments.}

\subsubsection{Edge computing with high mobility}
The works in \cite{mozaffari2017mobile,  wang2013energy,   jeong2017mobile, zhou2018uav,   hu2019uav, hu2019joint, tun2021energy, yang2020multiUAV, zhang2020computation, yang2019energy, huang2017distributed, liu2018computation, lee2019performance, kang2018blockchain, huang2017exploring} studied various problems related to edge computing in IoT networks that integrate highly mobile devices such as unmanned aerial vehicles (UAVs) and connected vehicles.  
First, in \cite{mozaffari2017mobile, wang2013energy, jeong2017mobile, zhou2018uav,   hu2019uav, hu2019joint, tun2021energy, yang2020multiUAV, zhang2020computation, yang2019energy}, the authors studied the use of UAVs for wireless and computing scenarios. %
For instance, the authors in \cite{mozaffari2017mobile} proposed a framework that jointly optimizes  UAV placement and uplink power control  so that UAVs can collect edge data from ground sensors. %
In \cite{wang2013energy}, the authors employed UAVs as edge message ferries that collect information in wireless sensor networks and carry the data to the destination. 
In  \cite{ jeong2017mobile, zhou2018uav,  hu2019uav, hu2019joint, tun2021energy, yang2020multiUAV, zhang2020computation, yang2019energy}, the authors proposed various use cases for deploying airborne edge computing using a UAV.  
In \cite{jeong2017mobile}, the authors investigated a UAV-mounted cloudlet in which UAVs equipped with a computing processor offload and compute the tasks offloaded from ground devices. %
The work in \cite{zhou2018uav} studied a UAV-enabled mobile edge computing system in which the users harvest the energy from the signal transmitted by the UAV in downlink, and the harvested energy is used to transmit in uplink. 
The work in \cite{hu2019uav} investigated a UAV-enabled edge computing system in which a UAV offloads computational tasks from users and decides whether to compute the tasks or transmit the tasks to a remote server. 
\textcolor{black}{In \cite{hu2019joint} and \cite{tun2021energy}, the authors proposed a UAV-aided multi-access edge computing (MEC) system in which a UAV acts as an edge server (or cloudlet) providing computation service for the ground devices. On the other hand, in \cite{yang2020multiUAV} and \cite{zhang2020computation}, multiple UAVs are assumed to act as edge computing devices which cooperatively compute tasks offloaded by ground devices. Also, the authors in \cite{yang2019energy} studied the joint problem of user association and computational task allocation in a mobile edge computing system where UAVs act as edge computing devices. Hence, the role of UAVs is changeable and determined depending on the considered network environment. In this paper, we focus on a scenario in which one of UAVs acts as a edge server and the rest of them act as edge computing devices. This scenario implies that the considered UAVs are not as powerful as a high performance computing server which can compute all tasks alone, however, they can compute a few tasks faster than other IoT devices such as sensors.}

%The authors in \cite{yang2019energy} studied the joint problem of user association and computational task allocation in a mobile edge computing system where UAVs act as edge computing devices. 
%\textcolor{blue}{In \cite{hu2019joint}, the authors proposed a UAV-aided mobile edge computing system where a UAV equipped with computing resources acts as an edge node for performing partially offloaded computing tasks from ground users. The work in \cite{tun2021energy} studied a UAV-aided mobile edge computing system in which the energy consumption at IoT devices and the UAVs are jointly minimized while considering the communication and computation latency requirements. The work in \cite{yang2020multiUAV} investigated a multi-UAV-aided mobile-edge computing system where multiple UAVs serve computation offloading services for ground IoT nodes while the computation loads of UAVs are balanced by a deployment mechanism. In \cite{zhang2020computation}, the authors studied a computation efficiency maximization problem in a multi-UAV assisted MEC system for jointly optimizing user-edge association, computation and radio resource allocation, and trajectory scheduling of UAVs based on the partial computation offloading.}

Next, edge computing is investigated in various scenarios incorporating connected vehicles  \cite{huang2017distributed, liu2018computation, lee2019performance, kang2018blockchain, huang2017exploring}. 
The authors in \cite{huang2017distributed} developed a distributed reputation management system in which the edge computing resources are allocated in a way to optimize security. 
The work in \cite{liu2018computation} proposed a low-complexity computation offloading algorithm that minimizes the computing cost at connected vehicles. 
Also, the work in \cite{lee2019performance} proposed the use of edge computing techniques to process the computational tasks required in a blockchain system by using the local computing resources of vehicular nodes. 
The authors in \cite{kang2018blockchain} developed a smart contract deployed on an edge computing system to enable connected vehicles to store and share the data securely. 
In \cite{huang2017exploring}, the authors applied a software-defined networking concept to develop an edge computing architecture in which the control plane protocol is designed to cluster a set of neighboring vehicles and a centralized edge computing server is used to optimize the data transmission path. 

\begin{table}[t!]
\caption{Comparison with related works in edge computing. (\checkmark: considered, -: not considered)}
\scriptsize
\centering
 \begin{tabular}{ p{3cm} p{2.2cm} p{1.6cm} p{2cm} p{2.2cm} p{1.6cm} } 
 \hline
  & Radio resource \newline allocation & Multiple edges & Edge mobility & Computation \newline heterogeneity & Time {constraints} \\
 \hline\hline
 \cite{wang2019qos} & - & - & - & - & - \\ 
 \cite{huang2019deep,mao2016dynamic,kuang2019partial,chen2018virtual} & \checkmark & - & - & - & - \\
 \cite{huang2017fair}  & - & - & - & \checkmark & - \\
 \textcolor{black}{\cite{wong2015improving, Yigitoglu2017foggy, yang2019multi}},\cite{wang2017computation} & \checkmark & \checkmark & - & - & -  \\
 \textcolor{black}{\cite{IBM}} & \textcolor{black}{\checkmark} & \textcolor{black}{\checkmark} & \textcolor{black}{-} & \textcolor{black}{\checkmark} & \textcolor{black}{-} \\
 \cite{mozaffari2017mobile},\cite{jeong2017mobile,zhou2018uav} & \checkmark & - & \checkmark & - & -  \\
 \cite{wang2013energy,hu2019uav} & - & - & \checkmark & - & -  \\
 \cite{yang2019energy} & \checkmark & \checkmark & \checkmark & - & -  \\
 \cite{hu2019joint, tun2021energy} & \checkmark & - & \checkmark & \checkmark & -  \\
 \cite{yang2020multiUAV, zhang2020computation} & \checkmark & \checkmark & \checkmark & \checkmark & -  \\
 \cite{huang2017distributed, liu2018computation, lee2019performance, kang2018blockchain, huang2017exploring} & - & \checkmark & - & - & -  \\
 Our work & \checkmark & \checkmark & \checkmark & \checkmark & \checkmark  \\
 \hline
 \end{tabular}
 \label{table:1}
\end{table}

\subsubsection{Limited time constraints within edge computing}

\textcolor{black}{The aforementioned prior works \cite{wang2019qos, huang2019deep, huang2017fair, mao2016dynamic, kuang2019partial, wang2017computation, chen2018virtual, mozaffari2017mobile,  wang2013energy,  jeong2017mobile, zhou2018uav,  hu2019uav, yang2019energy, hu2019joint, tun2021energy, yang2020multiUAV, zhang2020computation, huang2017distributed, liu2018computation, lee2019performance, kang2018blockchain, huang2017exploring} assume that edge computing operates during a relatively long time period, and they do not consider a constraint on the total edge computing time period. However, in IoT scenarios, edge computing can be initiated and discontinued at any time due to the completion of running an application or the mobility of the edge nodes such as drones and vehicles. To capture such use cases, we propose the concept of \emph{ephemeral edge computing} in which edge computing occurs among IoT devices that have a stringent time constraints within which they can perform edge computing. In Table \ref{table:1}, we provide a comprehensive comparison between our work and the existing works on computation offloading in edge computing.}

Next, we first provide the real-world examples of ephemeral edge computing scenarios and, then, we outline our key contributions in this area.

\vspace{-5mm}
\subsection{Ephemeral Edge Computing}

In real world applications, various edge devices can be used to form a local edge network spontaneously and process computational tasks of different applications. 
One common observation here is that the total time period is limited in real-world IoT examples. 
In particular, the running time of a local edge network can be limited due to mobility of edge devices.
Also, when edge computing is initiated to operate an IoT user's application, the usage time of the application can be finite.
Therefore, we introduce a notion of \emph{ephemeral edge computing} to capture cases in which edge computing occurs in a relatively short time period. 
\textcolor{black}{
Here, we note that there exists a suite of industry products related to edge computing (e.g., from Nokia or Amazon). However, these products are mostly related to infrastructure-based edge computing, and to our knowledge, they have not been yet exploited to deploy a concept such as ephemeral edge computing.
Meanwhile, the emerging O-RAN standard \cite{ORAN1}
will have capabilities to support short-lived computing transactions, however, O-RAN does not provide any ephemeral edge computing solution that can leverage these capabilities, as such solutions are left to the research community, which motivates the timeliness and need for this work.}
As discussed next, the concept of ephemeral edge computing admits many real-world IoT applications in several industrial and civilian areas in which total time period available for the use of edge computing  is constrained. 

\subsubsection{Intelligent transportation systems}

%\begin{figure}
%	\centering
%	\includegraphics[width=0.35\textwidth]{./figures/road.pdf}
%	\caption{ \small Ephemeral edge computing framework where the vehicles are edge devices that form a local edge computing network during a limited time period. The vehicles offload their computational tasks from sensors and process in a distributed way while moving on the same direction.}\vspace{-7mm}
%	\label{fig:road}
%\end{figure}

\begin{figure*}[t]\vspace{-2mm}
	\begin{multicols}{2}\vspace{-9mm}
		%\hspace{-4.5mm}
\centering
		\includegraphics[width=0.7\columnwidth]{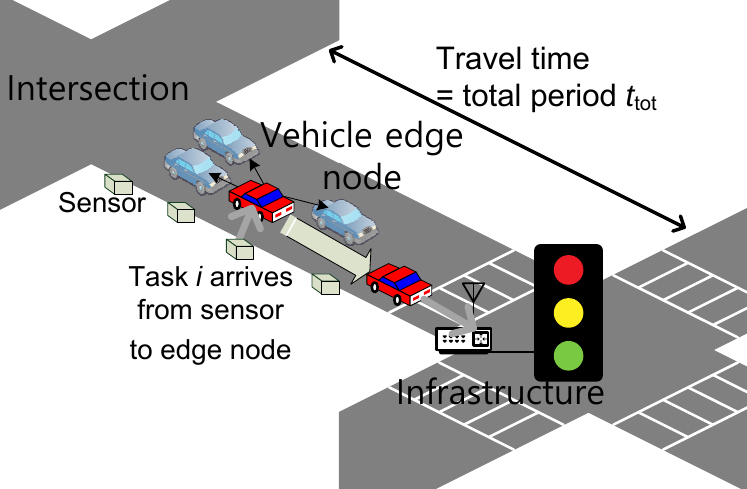}\vspace{-2.5mm}
		\par\caption{\small Illustrative example of ephemeral edge computing framework in intelligent transporation systems.}
\label{fig:road}
		%\hspace{-4.5mm}
		\includegraphics[width=1.05\columnwidth]{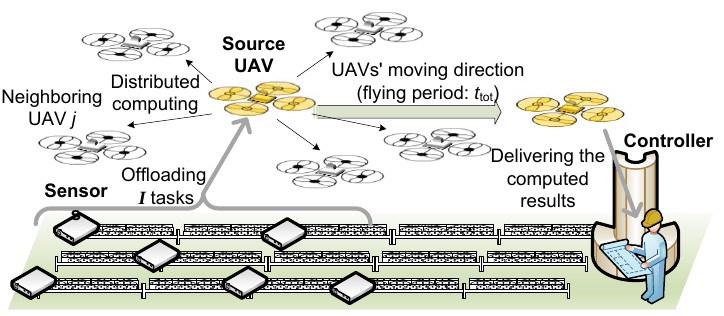}\vspace{-2.5mm}	
		\par\caption{\small Illustrative example of ephemeral edge computing framework in smart factory.}
\label{fig:factory}
		\end{multicols}\vspace{-10mm}
\end{figure*}

As shown in Fig.~\ref{fig:road}, edge computing can be applied to an urban road environment in which a number of sensors monitor the status of the road traffic, vehicle flow, and pedestrian generating a large data volume \cite{huang2017exploring}. 
For example, the generated sensory data from the road environment can be used to detect the current traffic status or to predict safety hazards. 
Moreover, the generated data can also be used to decide the signal light timing and schedule the vehicles at a merging ramp or intersection \cite{huang2017distributed}. 
Therefore, processing the sensory data from a road environment is essential to optimize and control the various physical components of transportation systems. 
In a road environment, since the road sensors have a low computational capability, edge computing on the vehicles can be used to offload the sensory data from environment. 
Then, the data is processed to extract meaningful information such as traffic forecast and safety warnings \cite{darwish2018fog, ferdowsi2019deep}. 
Once the data is processed by the vehicles' on-board computers, the vehicles can transmit the processed information to adjacent road side unit (RSU)  that can then use the processed information to control traffic flows. 

Therefore, intelligent transportation systems provide an important use case for ephemeral edge computing. 
In an urban environment such as the one shown in Fig.~\ref{fig:road}, a set of vehicles move from an intersection to the next intersection while maintaining a formation. 
When edge computing is implemented on the vehicles, it can only be maintained for a limited time period due to mobility. 
Those vehicles can cooperatively process the offloaded data within a limited time period that is the travel time between two intersections. 
Therefore, these vehicles will form an ephemeral edge computing network. 
In this case, the total time period dedicated to edge computing in a vehicular network will be affected by the vehicles' speed and trajectory. 
In particular, the vehicles can share the information on the destination and trajectory to estimate the time period during which a set of vehicles moving the same direction. 
This is just one example of edge computing among many others in the context of transportation systems. 

\subsubsection{Smart factory}

%	\begin{figure}
%	\centering
%	\includegraphics[width=0.55\textwidth]{./figures/factory.pdf}
%	\caption{ \small Ephemeral edge computing framework to offload  computational tasks from  sensors and allocate the offloaded tasks to neighboring edge-computing UAVs when the source edge-computing UAV is moving to the destination. }\vspace{-7mm}
%	\label{fig:factory}
%\end{figure}

	In emerging smart factory scenarios, also known as Industry 4.0 \cite{wang2016implementing}, sensors can detect malfunctions and send  diagnostics signals to actuators in the factory. 
Therefore, factory systems must be optimized to manage the process of sensory data transmission, low-latency computation, and proactive decision making in order to  quickly react to new situations \cite{zuehlke2010smartfactory}. 
Some key challenges for enabling the smart factory vision include effective in-network computing and improvement of wireless connectivity to integrate physical and digital systems, i.e., networking and computation. 
\emph{Computing sensory data in a timely manner} is essential to operate a physical factory system. 
To this end, the concept of ephemeral edge computing can be applied in cyber-physical smart factory systems where UAVs, robots, and drones are deployed and perform key functions such as data storage, computing, control, and transmission \cite{majd2016placement}. 

As shown in Fig.~\ref{fig:factory}, we consider a smart factory in which sensors monitor the status of the manufacturing process and generate a large data volume. 
For example,  the generated sensory data can be used as an input to  machine learning algorithms, e.g., for classification, to predict any abnormality in the manufacturing process. 
Hence, a number of computational tasks must be processed in order to make a decision on how to control the physical systems of the factory based on the information extracted from the data. 
However, due to the low computational capability of the sensors, it is not possible to compute those tasks locally at the sensors. 
Also, sensors are not able to transmit  data over a long distance, and, hence, a flexible relay is necessary \cite{jawhar2014framework}. 
For example, edge-enabled UAVs can be used in a smart factory to gather the tasks from the sensors, compute the tasks, and deliver the computed results to the destination, e.g., a central factory controller that can control the actuators. 
This is a meaningful use case of ephemeral edge computing in that the local edge network can be maintained until the UAVs arrive at the destination. 
Here, the total time period of ephemeral edge computing corresponds to the moving time from the source location to destination.

\subsubsection{IoT sensor systems for end users}
%	
%	\begin{figure}
%		\centering
%		\includegraphics[width=0.35\textwidth]{./figures/home.pdf}
%		\caption{ \small An IoT environment where the the sensory IoT data is instantly requested and processed by an ephemeral edge computing network to operate real-time applications such as augmented reality and gaming. }\vspace{-7mm}
%		\label{fig:home}
%	\end{figure}
	
%	As shown in Fig.~\ref{fig:home}, consider an IoT environment in which the generated sensory data from the IoT devices is used to control and monitor the status of home appliances, to detect a user's motion and voice \cite{schneider2017augmented}, or to run gaming and augmented reality applications at a museum, sport events, and  sightseeing places  \cite{hu2015mobile}. 	
%
%
	Consider an IoT environment in which the generated sensory data from the IoT devices is used to control and monitor the status of home appliances, to detect a user's motion and voice \cite{schneider2017augmented}, or to run gaming and augmented reality applications at a museum, sport events, and  sightseeing places  \cite{hu2015mobile}. 	
	Those applications require processing and analysis of the real-time IoT data. 	
	In particular, augmented reality and gaming applications must process the data depending on the user's location and orientation. 
	In this case, the time duration within which a user's device is at a stable location in space can be relatively short, and ephemeral edge computing is needed to process the IoT data in a limited time period. 

	As a result, the aforementioned examples in this section show that: a) Ephemeral edge computing admits a diverse set of IoT applications and b) in these applications, the time period dedicated to ephemeral edge computing can be limited depending on the various factors such as mobility and usage patterns of applications. 
	When the total time period of ephemeral edge computing is limited, there is a need for a new approaches to efficiently allocate the radio and computing resources to process a maximum number of computational tasks while considering the time-sensitive nature of the system.

\vspace{-5mm}
\subsection{Contributions}
\textcolor{black}{In all of these existing works on edge computing \cite{wong2015improving, Yigitoglu2017foggy, yang2019multi, IBM, wang2019qos, huang2019deep, huang2017fair, mao2016dynamic, kuang2019partial, wang2017computation, chen2018virtual, mozaffari2017mobile,  wang2013energy, jeong2017mobile, zhou2018uav,  hu2019uav, yang2019energy, hu2019joint, tun2021energy, yang2020multiUAV, zhang2020computation, huang2017distributed, liu2018computation, lee2019performance, kang2018blockchain, huang2017exploring}, it is generally assumed that edge computing is formed and used for a relatively long time period, and, therefore, the total computing time of edge computing is not considered. 
As shown in the real-world examples of ephemeral edge computing, edge computing can be initiated and discontinued at any time, resulting in the finite total time period to use edge computing. 
Therefore, we propose the concept of \emph{ephemeral edge computing} in which the total edge computing time is limited. 
Also, the prior art on edge computing employing both communications and computing \cite{wang2019qos, huang2019deep, huang2017fair, mao2016dynamic, kuang2019partial, wang2017computation, chen2018virtual, mozaffari2017mobile,  wang2013energy, jeong2017mobile, zhou2018uav,  hu2019uav, yang2019energy, hu2019joint, tun2021energy, yang2020multiUAV, zhang2020computation, huang2017distributed, liu2018computation, lee2019performance, kang2018blockchain, huang2017exploring}, generally assumes that information on prospective computing tasks such as data size and arriving order is completely known. 
However, in practice, the information on tasks can be revealed gradually over time since sensory data is randomly generated. 
Hence, when a series of tasks are offloaded to a neighboring edge node, predicting prospective future tasks is often not possible.}
Moreover, instead of offloading the computational tasks to  base stations that are connected the servers, as done in \cite{wang2019qos, huang2019deep, mao2016dynamic}, and \cite{huang2017distributed, kang2018blockchain, lee2019performance, huang2017exploring, liu2018computation}, the tasks can be offloaded to neighboring edge devices by using device-to-device (D2D) communications so as to reduce a communication latency. 
Furthermore, instead of relying on a single edge node for computing, as done in \cite{jeong2017mobile, zhou2018uav,   hu2019uav}, it is beneficial to leverage multiple, neighboring edge nodes for distributed computing of tasks.  
Consequently, unlike the existing literature \cite{wang2019qos, huang2019deep, huang2017fair, mao2016dynamic, kuang2019partial, wang2017computation, chen2018virtual, mozaffari2017mobile,  wang2013energy,  jeong2017mobile, zhou2018uav,  hu2019uav, yang2019energy, hu2019joint, tun2021energy, yang2020multiUAV, zhang2020computation, huang2017distributed, liu2018computation, lee2019performance, kang2018blockchain, huang2017exploring} which assumes full information knowledge on tasks and adopts either single edge node computing models or the models placing edge computing at the base stations, our goal is to design an \emph{online approach} to maximize the number of computed tasks on a network of multiple end-user edge nodes engaged in an ephemeral edge computing network in which there is a strict and  limited total edge-computing time, when the information on tasks is revealed in an online manner. 

The main contribution of this paper is a \emph{novel framework for distributed ephemeral edge computing} that can be operated within a limited time period, as needed in the applications of Figs. {\ref{fig:road}} and {\ref{fig:factory}}.
In particular, our framework allows tasks from sensors to be offloaded to a source edge node, which can subsequently allocate tasks to neighboring edge nodes for computation before the source node finishes edge computing. 
When the exact information on the offloaded tasks is unknown to the source node, it is challenging to decide which neighboring edge node has to compute which task.
\textcolor{black}{If a prior information on the task size is known to the source node, the computation delay at each neighboring edge node can be determined and the source node will allocate the tasks to the edge nodes according to their computational speed and the size of the tasks. However, in practice, the computational tasks arrive dynamically to the source edge node under a real-time process (i.e., online process) and their different data sizes cannot be known in advance.}
Therefore, we formulate an online optimization problem whose goal is to maximize the number of computed tasks when the total time period dedicated to ephemeral edge computing is constrained. 
To solve this problem without any prior information on the future task size, we propose a new online greedy algorithm that is used by the source edge node to make an on-the-fly decision for selecting one of the neighboring node upon the sequential arrival of the computational tasks while a prior information on the task size is unknown. 
Then, we analyze the performance of the proposed algorithm by using the notion of competitive ratio; defined as the ratio between the number of computed tasks achieved by the proposed algorithm and the optimal number of computed tasks that can be achieved by an offline algorithm. 
To this end, we apply the concept of primal-dual approach where the ratio between the dual problem and the original problem constitutes a competitive ratio.  
Therefore, we derive dual problem so as to analyze the worst-case performance of the proposed online algorithm. 
By doing so, the worst-case competitive ratio can be derived as a function of the task sizes and the  communication  and computing performance of the neighboring edge nodes. 
Simulation results show that the proposed online algorithm can maximize the number of computed tasks and achieve a performance that is near-optimal compared to an offline solution that has full information on tasks. 

The rest of this paper is organized as follows. In Section~\ref{sec:systemmodel}, we present the system model. 
Section~\ref{sec:problem} formulates the proposed online problem. 
Section~\ref{sec:algorithm} presents our proposed solution and performance analysis. 
Simulation results are analyzed in Section~\ref{sec:simulation} while conclusions are drawn in Section \ref{sec:conclusion}.

\vspace{-5mm}
\section{System Model and Problem Formulation}\label{sec:systemmodel} 
\vspace{-5mm}
\subsection{System Model}

We consider an ephemeral edge computing system in which sensors generate a set $\mathcal{I}$ of $I$ tasks\footnote{{\textcolor{black}{For consistency, we use the term "task" to indicate both the data generated by a sensor and the computational job that will be used to process data.}}} that are offloaded to a given edge node that we refer to hereinafter as the \emph{source edge node}. 
The source edge node can be seen as a node with mobility such as vehicles and UAVs. Also, the source edge node can be a static node.
\textcolor{black}{While the scenarios that can use ephemeral edge computing are diverse, the role of the source edge node is to offload the computational task data from the sensors and allocate them to neighboring edge nodes. Then, each neighboring edge node directly delivers the computed result to the destination, such as a central controller in a smart factory or an RSU in intelligent transportation systems. Finally, the destination collects the computed tasks from the neighboring edge nodes and makes a decision on how to control the physical systems of the factory based on the collected data.}
When tasks reach the source edge node, they are labeled by their order of arrival. 
Thus, a task that arrives a time instant $i$ is denoted as task $i\in \mathcal{I}.$
Since the source edge node processes the tasks using a first-input-first-output policy, it will sequentially compute its tasks. 
\textcolor{black}{The set $\mathcal{J}$ denotes the set of $J$ edge nodes that are neighbors to the source.}
Each edge node $j \in \mathcal{J}$ is used to compute some allocated task $i$ from the source edge node. 
We also consider that the set of neighboring edge nodes $\cJ$ is initially selected by the source edge node. 
In this regard, the source edge node selects the neighboring edge nodes that are moving towards its same destination. 
\textcolor{black}{Note that the term ``one task'' used here can be seen as a reference to a bundle of small tasks,
making it possible to execute multiple tasks at each edge node. Furthermore, the set $\mathcal{J}$ can include multiple virtual entities of an actual edge node when the number of edge nodes is too small to accommodate all of the tasks. In this case, $|\mathcal{J}|=kJ$ where $k$ is the number of virtual entities and $J$ is the number of actual edge nodes. The virtual entities of an actual edge node would then have to share the decision variables to have the same priority, and the edge nodes would execute their tasks in a round-robin manner.}
\textcolor{black}{The association between the source edge node and neighboring edge nodes can be established based on the clustering algorithm proposed in \cite{cooper2017comparative}, in which the cluster, cluster head, and cluster members correspond, respectively, to the ephemeral edge computing system, source edge node, and neighboring edge nodes. In the considered clustering algorithm, the source edge nodes exchange their link information, such as link states, computation capacity, and mobility, and each source edge node selects the qualified neighboring edge nodes that can maintain the connectivity during $t_{\rm{tot}}$ with no computation in progress, based on the exchanged information.}
\textcolor{black}{The source edge node is assumed to select $J$ neighboring edge nodes that are qualified to join a local edge computing network to process the computational tasks
in terms of residual battery level and computation speed.}
\textcolor{black}{If there are no edge nodes (i.e., $J=0$), a sensor computes its tasks by itself and transmits the results to the destination.
In this paper, we use the edge node essentially for boosting the computation speed rather than for carrying data between sensors and controllers.
Note that, the case in which the node is static, can easily be accommodate into our framework.
For instance, a static source edge node offloads the computational task data from the sensors
and allocates them to neighboring static edge nodes.
Then, each static neighboring edge node calculates the allocated task and directly transmits the result to a destination.
Moreover, mobile edge nodes can be dispatched to any location such as mountains and rural areas where the fixed infrastructure is not readily accessible.}

%\footnotetext{\textcolor{blue}{If a task cannot be processed by an ephemeral edge computing system due to the expiration of time deadline or the dispersal of the edge nodes, the sensor should process the task and deliver the result to the destination by itself. Since conducting a task at a sensor leads to large computation latency and excessive energy consumption of battery-powered sensors, a sensor exploits an edge-computing system whenever it is available.}}
%

\begin{figure}
	\centering
	\includegraphics[width=0.65\textwidth]{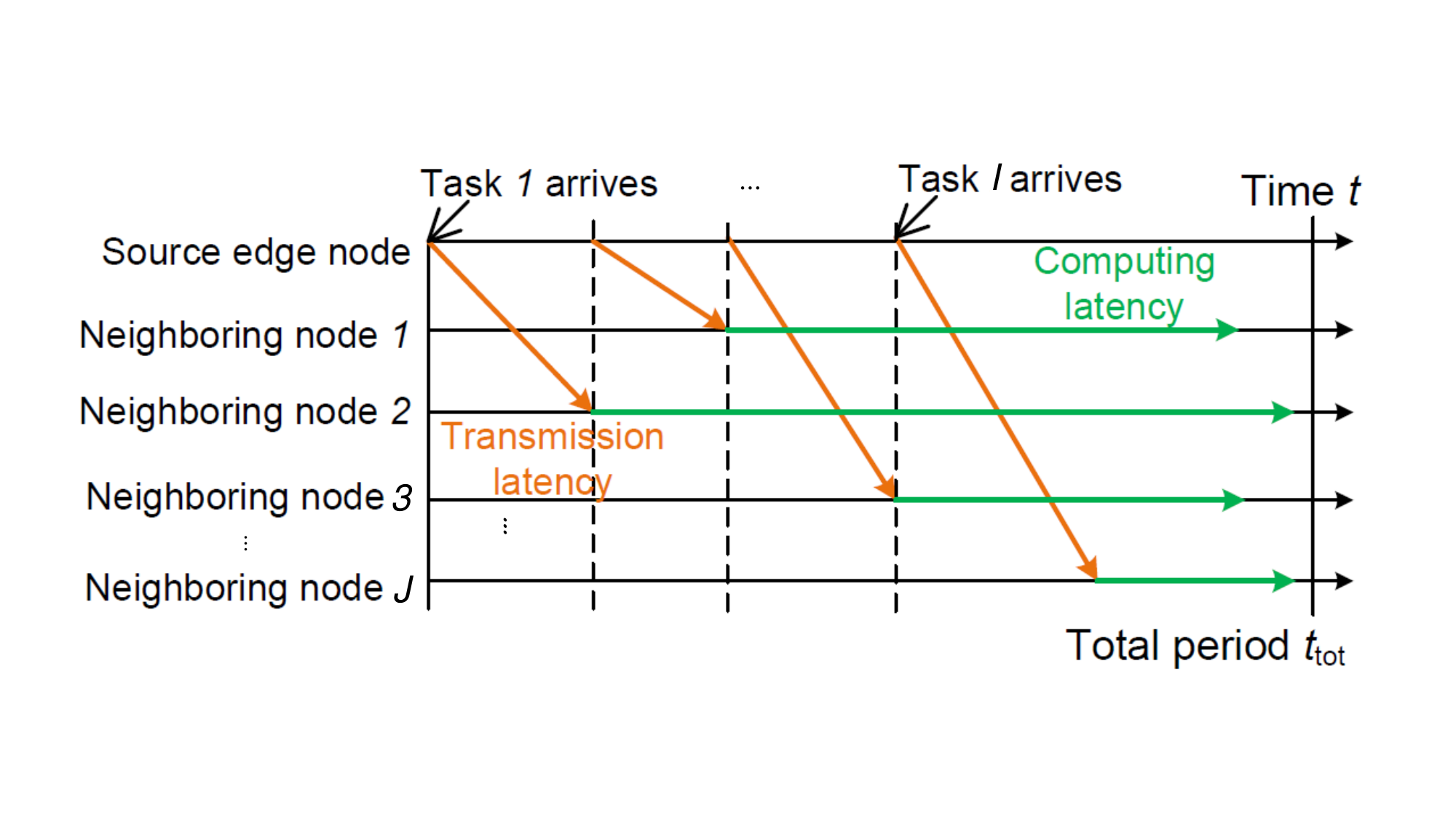}\vspace{-12mm}
	\caption{\small Online  edge computing framework to offload  computational tasks and allocate the offloaded tasks to neighboring edge nodes in total edge computing period $t_{\textrm{tot}}$ within an ephemeral edge computing system. }\vspace{-7mm}
	\label{fig:systemmodel}
\end{figure}

The source edge node allocates the computational tasks to other neighboring edge nodes. Such distributed computing can reduce the overall computational latency when multiple tasks are computed. 
Also, to prevent an excessive energy consumption at neighboring edge nodes, we assume that only one task is allocated to one edge node. 
Therefore,  when neighboring edge node $j$ computes task $i$, the decision variable is set as $y_{ij}=1$. 
The other edge nodes are not used to process the same task $i$, i.e., \textcolor{black}{if $y_{ij}=1$ then $y_{ij'}=0, \forall j' \in \cJ\setminus \{j\}$, $\forall i \in \cI$.} 
Task allocation to neighboring edge node incurs a \emph{transmission latency}. 
\textcolor{black}{The data rate pertaining to the transmission of the data of task $i$ to neighboring edge node $j$ will~be: $r_j = F(B, g_j, P_{t}, \sigma)$,}
%\begin{equation}
%    \textcolor{blue}{r_j = F(B, g_j, P_{t}, \sigma),}
%  %r_j = B \log_2 \left(1+\frac{g_j P_{t}}{\sigma^2}\right),
%\end{equation}
%%
\textcolor{black}{where $F$ is a general transmission rate function, $P_t$ is the transmit power of the source edge node, $B$ is the bandwidth, $\sigma^2$ is the noise power, and $g_j$ is the channel gain between the source edge node and neighboring edge node $j$.} %\textcolor{blue}{$F$ should be replaced by a proper channel capacity function according to the network environment (e.g., Shannon capacity).} %
Therefore, when the data size of task $i$ is $d_i$~bits, the transmission latency becomes ${d_i}/{r_j}$.
\textcolor{black}{Once task $i$ is received by neighboring edge node $j$, it will be processed within a \emph{computational latency}}\footnote{\textcolor{black}{One way to estimate computation latency is to define how many CPU cycles are needed for computing a bit of data. In this paper, the computation latency is defined as $\alpha \cdot d_i/F_j$ where $\alpha$ is the required number of CPU cycles per bit (i.e., computation complexity) and $F_j$ is the CPU speed of edge node $j$ in Hz. For notational simplicity, we introduce the computation power as $f_j = F_j/\alpha$. Therefore, the computation delay can be simply expressed as $d_i/f_j$.}} \textcolor{black}{${d_i}/{f_j}$ where $f_j$ is the computation power of edge node $j$.}

In the proposed ephemeral edge computing system, the time period that the source edge node actively \textcolor{black}{uses} edge computing is given by $t_{\textrm{tot}}$. 
To determine $t_{\textrm{tot}}$, key features of the edge nodes can be considered. 
From the aforementioned cases of ephemeral edge computing, the total time period $t_{\textrm{tot}}$ can be determined as the moving time period of a set of edge computing vehicles on the road or UAVs in a smart factory. 
For example, $t_{\textrm{tot}}$ can depend on the mobility that is characterized by the speed and moving distance of the source edge node. 
$t_{\textrm{tot}}$ could also depend on the different trajectories of the source edge node. 
In the IoT scenarios, the total time period can be given as the running time of an application or the time period where a smart device is staying near the other edge devices to deploy an edge computing network. 
For a given  $t_{\textrm{tot}}$, if a certain number of tasks is processed as shown in Fig.~\ref{fig:systemmodel}, then the tasks' transmission and computation  must be completed within $t_{\textrm{tot}}$. \textcolor{black}{We define the duration between the arrival of the first task and completion of task $i$ as the \emph{completion time} of task $i$.}
As shown in Fig.~\ref{fig:systemmodel}, the tasks are sequentially offloaded from the source edge node to one of neighboring nodes. 
For instance, when the first task, $i=1$, is being allocated, the \textcolor{black}{completion time} including transmission and computation of task 1 will be:%
\begin{equation}\label{const1}
	\sum_{j=1}^J  d_1\left( \frac{1}{r_j} + \frac{1}{f_j}\right)   y_{1j} \leq t_{\textrm{tot}}. 
\end{equation}
Subsequently, since  tasks are  sequentially transmitted to the neighboring edge nodes in the order of index $i$, there will be $i-1$ transmissions before task $i$ is transmitted. 
Therefore, the \textcolor{black}{completion time of} any task $i$, $\forall i \in \mathcal{I} \setminus \{1\}$, %
\begin{equation}\label{const2}
	\sum_{i'=1}^{i-1} \sum_{j=1}^J   d_{i'}\left( \frac{1}{r_j}\right)  y_{i'j} + \sum_{j=1}^J d_i \left( \frac{1}{r_j} + \frac{1}{f_j}\right)  y_{ij} \leq t_{\textrm{tot}},
\end{equation}
where the first term is the sum of the transmission latency of $i-1$ tasks, and the second term is the transmission and computation \textcolor{black}{latency} of task $i$. 
\textcolor{black}{Given that tasks are allocated to and computed by neighboring edge nodes in the order of index $i$, the completion time of task $i$ in (\ref{const2}) includes the summation of the transmission latency for the previous $i-1$ tasks.}
%if task $i$ is computed within period $t_{\textrm{tot}}$, then  $i-1$ tasks will also be computed in the given period. 
Next, we formulate an online task allocation problem to study how tasks are distributed within an edge computing network. 

\vspace{-5mm}
\subsection{Problem Formulation}\label{sec:problem}

Our goal is to allocate tasks to neighboring edge nodes in order to complete the maximum number of tasks during the period $t_{\textrm{tot}}$ needed for the source edge node to reach its destination. 
To compute the tasks, the source edge node must allocate each task to a neighboring edge node that can yield low latency. 
In practice, when the computational tasks arrive dynamically to the source edge node, their different data sizes cannot be known in advance. 
As a result, the source edge node will be unable to know a priori the information on future tasks, and, therefore, optimizing the task distribution process under this uncertainty is very challenging. 
Under such uncertainty, selecting a neighboring edge node that computes a current task must also account for potential arrival of future tasks. 
When the future information is revealed sequentially, the arrival of information can be captured within an online optimization framework. 
In particular, by using online optimization techniques such as those in \cite{buchbinder2009design}, it is possible to make an on-the-fly decision while the future information is given in an online manner. 
To cope with the uncertainty of the future task arrivals while considering the data rate and computing capabilities of given neighboring edge nodes, we will thus propose a rigorous \emph{online optimization framework} that can handle the problem of task allocation under uncertainty. 

First, we formulate the following online task allocation problem whose goal is to maximize the number of computed tasks when the total latency is limited by $t_{\textrm{tot}}$: \vspace{-3mm}
\begin{eqnarray}\label{problemD}
  \textrm{(D)}: \max_{\boldsymbol{y}} && \sum\nolimits_{i=1}^{I} \sum\nolimits_{j=1}^J y_{ij}\\
  \textrm{s.t.} && \eqref{const1}, \eqref{const2},\nonumber\\
  && \sum\nolimits_{i=1}^I y_{ij} \leq 1, \forall j \in \cJ, \label{const3}\\
  && \sum\nolimits_{j=1}^J y_{1j} \leq 1,  \label{const4}\\
  && \sum\nolimits_{j=1}^J (-{\color{black}y_{i-1j}}+y_{ij}) \leq 0, \forall i \in \cI \setminus \{1\}.  \label{const5}
\end{eqnarray}
where $\boldsymbol{y}$ is the vector of decision variables $y_{ij}, \forall i \in \cI, \forall j \in \cJ$. 
Hereinafter, this problem is called the dual problem. 
\textcolor{black}{Constraints \eqref{const1} and \eqref{const2} show that task $i$'s completion time must be smaller than $t_{\textrm{tot}}$ and tasks that cannot satisfy those constraints will not be offloaded.}
\textcolor{black}{\eqref{const3} implies that each neighboring edge node can compute at most one task to prevent an excessive energy consumption at any given edge node.}
In constraint \eqref{const4}, the first task is allocated to one of the neighboring edge nodes. 
Constraint \eqref{const5} implies that task $i$ can be allocated to a neighboring edge node if the task allocation of task $i-1$ is successful, i.e., $\sum_{j=1}^J {\color{black}y_{i-1j}} =1$. Otherwise, if $\sum_{j=1}^J {\color{black}y_{i-1j}} =0$, then, task $i$ cannot be allocated to any edge node, and $\sum_{j=1}^J y_{ij} = 0$. 
Due to \eqref{const4} and \eqref{const5}, we have $\sum_{j=1}^J y_{ij} \leq 1, \forall i \in \cI$, and, thus, each task is allocated to only one of neighboring edge nodes.
\textcolor{black}{Given that the demand for mobile device has been growing
exponentially in recent years, mainly driven by various emerging IoT applications,
we assume that $J \ge I$ and all tasks can be completed during a limited time $t_{\rm{tot}}$ using the edge computing network, and thus, there exist some feasible solutions that satisfy all the constraints in problem (D).}
%\textcolor{blue}{Given that the demand for mobile device has been growing exponentially in recent years, mainly driven by various emerging IoT applications, we assume that $J \ge I$.
%Note that the proposed framework can be repeatedly run so that the sensors can offload their tasks over time to the edge nodes if $J < I$.
%Instead of this delay-tolerant application, however, we focus on a real-time application where all tasks can be completed during a limited time $t_{\rm{tot}}$ using the edge computing network.
%In addition, the assumption of $J \ge I$ makes problem (D) to have different feasible solutions so that the performance of the proposed online algorithm can be measured.}
%, and thus, there exist some feasible solutions that satisfy all the constraints in problem (D).

Note that problem (D) is an \emph{online optimization problem}  and is challenging to solve using conventional offline approaches.  
This is because the value of $d_i, \forall i$, is sequentially revealed. 
When the tasks that  different sensors send to the source edge node  have a random size, the arrival sequence of $d_i$ is assumed to be unpredictable and unknown. 
At the moment when $d_i$ is disclosed, the source edge node knows only the current and past tasks.
However, the source edge node must make an \emph{instant and irrevocable online decision} on which neighboring edge node will compute task $i$. 
Under such uncertainty on $d_i$, allocating tasks to existing neighboring edge nodes must also account for potential arrival of new tasks. 
In fact, even if a given task allocation can compute an existing task successfully, it may have a detrimental effect on the allocation of  incoming tasks. 
In particular, if an edge node having a high data rate and high computational speed is already assigned to compute a previous task, it may not be possible to compute a future task having a large size. 
Therefore, it is challenging to optimize the task allocation  between incoming tasks and neighboring edge nodes.

In an online setting, the ad-auction problem in \cite{buchbinder2009design} shows a generalized structure of  an online linear programming problem and its algorithmic solution. 
We observe that the ad-auction problem and our problem have a key difference in the dependency of the constraints. 
In particular, the ad-auction problem includes the independent constraints about the maximum allocation size for each buyer that corresponds to the edge node in our problem. 
However, in our problem, the constraints about the maximum allocation size of edge nodes are dependent on each other. 
For instance, in \eqref{const1} and \eqref{const2}, the sum of the transmission latency of the previous tasks and the processing latency of the current \textcolor{black}{task} should be less than $t_{\textrm{tot}}$. 
The total time period is a function of the task allocation decisions of all edge nodes while each edge node has an independent task allocation size. 
Therefore, if the given budget of total time period is previously spent to offload and compute previous tasks, the source node cannot offload a new task to a neighboring node that is still available to accept a task. 
Additionally, our problem assumes that the arriving tasks are sequentially allocated to the neighboring node. 
For instance, the current task cannot be allocated to any node, if the previous task is not allocated due to constraints \eqref{const4} and \eqref{const5}. 
Due to the aforementioned differences, we need to develop a novel online task allocation strategy to solve  problem  (D). %

\vspace{-5mm}
\section{Proposed Online Task Allocation Framework}\label{sec:algorithm}

Our goal is to determine the vector of decision variables $\boldsymbol{y}$ so that the maximum number of sequentially arriving tasks is successfully computed by our distributed ephemeral edge computing system.  
When task size $d_i$ is unpredictable, the decision is not trivial since the current decision may affect the task allocation of future tasks, and all tasks cannot be computed due to the limited time resource $t_{\textrm{tot}}$. 
In this case, making an on-the-fly online decision, can process a smaller number of tasks than that of offline decision in which the complete information on all tasks is initially known. 
Therefore, the gap between the results achieved by online and offline cases must be minimized. 
To this end, the notion of \emph{competitive ratio} \cite{buchbinder2009design} from competitive analysis can be used to measure the performance of our online algorithm. 
It is an effective metric that compares the ratio between the the objective function's value achieved by an online algorithm and that of the offline optimal solution. 
In particular, \textcolor{black}{the upper bound of the competitive ratio} can be defined as a constant $\gamma$ such that  \vspace{-5mm}
\begin{eqnarray}\label{defcr}
 1\leq  \frac{\textrm{D}_\textrm{IP,OPT}}{\textrm{D}_{\textrm{IP}}} \leq \gamma, \vspace{-3mm}
\end{eqnarray}
where $\textrm{D}_\textrm{IP,OPT}$ denotes the offline optimal solution (OPT) of problem (D) in the form of integer programming (IP), i.e., the maximum number of computed tasks with the integer solution of $y_{ij}$. 
We will measure the performance of our proposed algorithm by observing the upper bound value defined by $\gamma$. 

To find the upper bound of problem (D), we use the structure of the primal and dual approach \cite{buchbinder2009design}. 
To this end, the optimization variables $y_{ij}$ are relaxed to be linear, i.e., $y_{ij}\in[0,1]$. 
By using the duality of linear programming, problem (D) can be rewritten as: 
\begin{eqnarray}\label{problemP}
\textrm{(P)}:	\min_{\boldsymbol{x}, \boldsymbol{z}, u_1}\hspace{-4mm}&& \sum\nolimits_{i=1}^I t_{\textrm{tot}} x_i + \sum\nolimits_{j=1}^J z_j + u_1,\\
	\textrm{s.t.}&& \hspace{-3mm}
			 \left( \frac{1}{r_j} + \frac{1}{f_j} \right) d_i x_i + \left(\frac{d_i}{r_j} \right)  \sum_{i'=i+1}^I x_{i'} + z_j  + u_i - u_{i+1} \geq 1, \forall i \in \cI \setminus \{I\}, \forall \! j \! \in \cJ, \label{const6}\\
	 &&	\hspace{-3mm} \left( \frac{1}{r_j} + \frac{1}{f_j} \right) d_I x_I + z_j + u_I \geq 1, \forall j \in \cJ, \label{const7}\\
	 && x_i \geq 0, z_j \geq 0, u_i \geq 0,
\end{eqnarray} 
where $\boldsymbol{x}$ and $\boldsymbol{z}$ are vectors with elements $x_{i}, \forall i \in \cI,$ and $z_{j}, \forall j \in \cJ$, respectively. 
This problem is called the \emph{primal problem}. 
In problem \eqref{problemP}, $x_1$, $x_{i\geq 2}$, $z_j$, $u_1$ and $u_{i\geq 2}$ are the dual variables associated, respectively, with constraints \eqref{const1}, \eqref{const2}, \eqref{const3}, \eqref{const4}, and \eqref{const5} in  problem (D). 

The values of \eqref{problemD} and \eqref{problemP} are denoted by \textcolor{black}{$\textrm{D}_{\textrm{IP}}$ and $\textrm{P}_{\textrm{LP}}$}, respectively. 
With $\textrm{D}_{\textrm{IP}}$ and $\textrm{P}_{\textrm{LP}}$, a \emph{competitive ratio} in \eqref{defcr} is derived. 
From the dual and primal problem formulation, it can be shown that $	\textrm{D}_{\textrm{IP}} \leq \textrm{D}_{\textrm{LP}} \leq \textrm{D}_{\textrm{LP,OPT}} \leq \textrm{P}_{\textrm{LP,OPT}} \leq \textrm{P}_{\textrm{LP}}.$
The first inequality is due to the fact that a linear relaxation allows problem (D), which is in the form of linear programming (LP), to have a higher value. 
The second inequality indicates that the offline optimal solution always achieves a value higher than or equal to the online solution of problem (D).  
The third inequality captures the slackness of the primal and dual problems. 
\textcolor{black}{In the fourth inequality, the offline optimal solution of problem (P), i.e., $\textrm{P}_{\textrm{LP,OPT}}$ is smaller than or equal to any online solution of problem (P), i.e., $\textrm{P}_{\textrm{LP}}$.}
%In the fourth inequality, the offline optimal solution of problem (P) is smaller than or equal to any online solution.   
Also, we have $	\textrm{D}_{\textrm{IP}} \leq \textrm{D}_{\textrm{IP,OPT}} \leq \textrm{D}_{\textrm{LP,OPT}}.$
The first inequality follows from the optimality gap between the online and offline solutions when $y_{ij}$ is an integer.  
The second inequality shows that linear relaxation of $y_{ij}$ allows us to have a higher value in problem (D). 
Thus,  the ratio in \eqref{defcr} becomes: $\frac{\textrm{D}_{\textrm{IP,OPT}}}{\textrm{D}_{\textrm{IP}}} \leq \frac{\textrm{P}_{\textrm{LP}}}{\textrm{D}_{\textrm{IP}}}$,
%\begin{eqnarray}\label{defcr2}
%\frac{\textrm{D}_{\textrm{IP,OPT}}}{\textrm{D}_{\textrm{IP}}} \leq \frac{\textrm{P}_{\textrm{LP}}}{\textrm{D}_{\textrm{IP}}},
%\end{eqnarray}
%
where $\textrm{P}_{\textrm{LP}}/ \textrm{D}_{\textrm{IP}}$ corresponds to $\gamma$ in \eqref{defcr}. 
Therefore, $\textrm{P}_{\textrm{LP}}/ \textrm{D}_{\textrm{IP}}$ becomes \textcolor{black}{the upper bound of the competitive ratio}.

\begin{algorithm}[t]\smallskip
\caption{Online Task Allocation Algorithm}
\begin{algorithmic}[1]\label{algorithm}
\scriptsize
\item[1 :] \hspace{0.0cm}  Initialize $y_{ij}=x_i=z_j=u_i=0, \forall i, j$. %
\item[2 :] \hspace{0.0cm}  {\bf for} $i \in \cI$ 
\item[3 :] \hspace{0.3cm}		Task $i$ arrives at source node. 
\item[4 :] \hspace{0.3cm}		Select edge node by using \eqref{selection}. %
\item[5 :] \hspace{0.3cm}		{\bf if} \eqref{const1} and \eqref{const2} are satisfied, and $\sum_j y_{i-1j} =1,$
\item[6 :] \hspace{0.6cm}			$y_{ij} \leftarrow 1$.
\item[7 :] \hspace{0.6cm}			Allocate task $i$ to edge node $j_i^*$ defined in \eqref{selection}. 
\item[8 :] \hspace{0.6cm}			Update $z_j$,  $x_i$, and $u_i$, respectively, by using \eqref{updatez}, \eqref{updatex}, and \eqref{updateu}. 
\item[9 :] \hspace{0.3cm}		{\bf otherwise},
\item[10:] \hspace{0.6cm}			$y_{ij} \leftarrow 0$.
\item[11:] \hspace{0.6cm}			Set $\Delta u_i=1$ and update $u_{i'}, \forall i' \leq i$%
\item[12:] \hspace{0.3cm}		{\bf end if} 
\item[13:] \hspace{0.0cm}  {\bf end for} 
\end{algorithmic}
\end{algorithm}

\vspace{-5mm}
\subsection{Online Greedy Algorithm}
To find the ratio $\textrm{P}_{\textrm{LP}}/ \textrm{D}_{\textrm{IP}}$, we develop a  new online greedy algorithm (Algorithm~\ref{algorithm}) specifically designed to solve problems (D) and (P), based on a general online optimization framework using the primal and dual approach of \cite{buchbinder2009design}. %
In Algorithm~\ref{algorithm}, the decision variables $y_{ij}$, $x_i$, $z_j$, and $u_i$ are updated while observing the new value of $d_i$. 
In particular, when task $i$ arrives to the source edge node, the original dual problem is solved by determining the value of $y_{ij}$.  
Also, other dual variables $x_i$, $z_j$, and $u_i$ are updated in order to find the performance bound of the proposed online algorithm. 
At the initial step of Algorithm~\ref{algorithm}, all variables are set to 0. 
The algorithm selects which edge node should compute task $i$. 
Since it is beneficial to offload task $i$ from the source node to the neighbor with a high data rate and computing speed, this decision rule can be designed to select an edge node with the shortest communication and computing latency to process the task $i$. 
To this end, edge node $j^*$ is selected by following the decision rule:
\begin{equation}\label{selection}
	j^* = \argmax_{\forall j}  \frac{(1-z_j)^\alpha}{\left(\frac{1}{r_j} + \frac{1}{f_j}\right) d_i},
\end{equation}
where $\alpha \geq 1$ is a constant used to guarantee that at least one of the tasks can be fairly allocated among the neighboring edge nodes. 
\textcolor{black}{The detailed derivation of the decision rule in \eqref{selection} is presented in Appendix A.}
Since $z_j$ is initially zero, the decision rule in \eqref{selection} only considers the latency required to process task $i$. 
In Algorithm~\ref{algorithm}, if a neighboring node $j$ accepts a task, the value of $z_j$ is updated to become positive. 
By doing so, $1-z_j$ is reduced, and, hence, another neighboring node can be selected when the next task arrives. 
However, if the neighboring node $j$ results in ${(1-z_j)^\alpha}/{\left(\frac{1}{r_j} + \frac{1}{f_j}\right) d_i} \geq {(1-z_{j'})^\alpha}/{\left(\frac{1}{r_{j'}} + \frac{1}{f_{j'}}\right) d_i}, \forall j' \in \cJ \setminus \{j\}$, the same node $j$ can be selected again. 
This can violates constraint \eqref{const3} that restricts each neighbor to accept one task. 
Therefore, a large value of $\alpha$ can be used to make $(1-z_j)^\alpha$ close to zero. 
Then, at the arrival of a new task, a different neighboring node is selected as $j^*$ by using decision rule \eqref{selection}. 

After a neighboring node $j^*$ is selected for task $i$, if the time budget is still available for the current task $i$ from constraints \eqref{const1} and \eqref{const2}, neighbor node $j^*$  finally receives task $i$ from the source node and performs processing. 
At this moment, the dual and primal variables are updated in Algorithm~\ref{algorithm}. 
The algorithm sets $y_{ij^*}=1$ showing that task  $i$ is allocated to edge node $j^*$. 
Next, the value of $z_{j^*}$ \textcolor{black}{must} be updated since $z_{j^*}$ is the primal variable associated with the dual problem's constraint \eqref{const3} with $j={j^*}$. 
When a neighboring node initially does not have any accepted task, all $z_j, \forall j \in \cJ$ are set to zero. 
However, if a neighboring node $j$ accepts a task $i$, $z_j$ will be updated as follows:
\begin{equation}\label{updatez}
	z_j = z_j \left(1+\left(\frac{1}{r_j} + \frac{1}{f_j}\right) \frac{d_i}{ t_{\textrm{tot}} }  \right) 
	+\left(\frac{1}{r_j} + \frac{1}{f_j}\right)  \frac{d_i}{t_{\textrm{tot}}}  \left(\frac{1}{c-1}\right),
\end{equation}
where $c>1$ is a positive constant that will be defined later. 
Also, the total time period $t_{\textrm{tot}}$ is assumed to be enough to process at least one task, and, thus, $\left(\frac{1}{r_j} + \frac{1}{f_j}\right) \frac{d_i}{ t_{\textrm{tot}} } < 1$. 
Meanwhile, the update of $x_i$  must satisfy constraints \eqref{const6} and \eqref{const7}. 
The value of $x_i$ is updated by using the rule:\vspace{-3mm}
\begin{equation}\label{updatex}
	x_i =  \frac{(1-z_j)^\alpha}{\left({1}/{r_j} + {1}/{f_j}\right) d_i}.
\end{equation}
Moreover, the values of $u_{i'}, \forall i' \leq i,$ is updated as follows: \vspace{-3mm}
\begin{equation}\label{updateu}
u_{i'} = u_{i'} + \Delta u_i, \forall i' \leq i, \vspace{-3mm}
\end{equation}
where we define, $\forall j' \in \cJ$,\vspace{-3mm}
\begin{eqnarray}
\Delta u_i 
&\triangleq&  \max_{j' \in \cJ \setminus \{j^* \}} \left( 1-  \left(	\left( \frac{1}{r_{j'}} + \frac{1}{f_{j'}} \right)   \frac{(1-z_{j^*})^\alpha}{\left({1}/{r_{j^*}} + {1}/{f_{j^*}}\right) } + z_{j'}  \right) , 0 \right).
\end{eqnarray}
\textcolor{black}{Otherwise, if the edge nodes in $\cJ$ do not satisfy \eqref{const1} and \eqref{const2}, then, the tasks arriving after task $i$ cannot be computed, i.e., $y_{ij}=0$,
and those tasks will not be offloaded.}
In this case, to satisfy constraints \eqref{const6} and \eqref{const7}, Algorithm~\ref{algorithm} updates any $z_j$ that has a value of 0 to 1 if $J\leq I$, or, otherwise, $\Delta u_i$ is set to 1. 
This update is intended to satisfy the constraints \eqref{const6} and \eqref{const7} for all $i\in\cI$ and $j\in\cJ$.  
For the arrival of each task, the proposed algorithm is a one-shot decision making process to find a feasible solution. Therefore, by iterating the proposed algorithm for all arriving tasks during $t_{\textrm{tot}}$, our algorithm converges to a feasible solution of problem (D). 

\vspace{-5mm}
\subsection{Performance Analysis}
For the analysis hereinafter, we assume that $\alpha=1$ for analytical tractability. 
In practice, this assumption implies that the decision rule \eqref{selection} tends to select the neighboring node with a high data rate and computing speed. 
As $\alpha$ increases, the decision rule selects a new neighboring node that has not been used to process any previous task. 
Now, as a first step to derive the competitive ratio of the proposed algorithm, we find the following result. 
\begin{lemma}\label{lemma1}
	The constraints of the primal problem \eqref{const6} and \eqref{const7}  will be satisfied if $z_j$, $x_i$, and $u_i$ are updated by \eqref{updatex}, \eqref{updatez}, and \eqref{updateu}, respectively. 
\end{lemma}\vspace{-3mm}
\begin{proof}
	See Appendix~\ref{appendix1}. 
\end{proof}
\noindent The next step of our analysis is to check whether the constraints in  problem (D) is satisfied. 
In particular, since it is observable that the upper bound of the left-hand side of the constraint \eqref{const3} can be greater than one,   \eqref{const3}  is not satisfied for $\alpha=1$, as shown next.  
\begin{lemma}\label{lemma2}
	In  \eqref{const3}, $\sum y_{ij}$ is  violated by at least 2. %
\end{lemma}\vspace{-3mm}
\begin{proof}
	See Appendix~\ref{appendix2}. 
\end{proof}
\noindent %
This result implies that more than two tasks can be offloaded to the same neighboring node. 
However, there exists a condition under which constraint \eqref{const3} is satisfied. 
\begin{lemma}\label{lemma2a}
	\eqref{const3} is  satisfied if $d_i > \left( {\left({1}/{r_{j_{i}^*}} + {1}/{f_{j_{i}^*}}\right)^{-1}} - {\left({1}/{r_{j_{I}^*}} + {1}/{f_{j_{I}^*}}\right)^{-1}} \right) t_{\textrm{tot}}(c-1)$ where $j_{i}^*$ is the  node selected to process task $i, \forall i \in \cI$. 
\end{lemma}\vspace{-3mm}
\begin{proof}
	After task $i$ is offloaded to node $j_i^*$, Algorithm~\ref{algorithm} updates $z_{j_{i}^*}= \left({1}/{r_{j_{i}^*}} + {1}/{f_{j_{i}^*}}\right) \frac{d_i}{t_{\textrm{tot}} (c-1)}$. 
	Next, when task $i+1$ arrives, the condition above yields the inequality $	\frac{1}{\left(		{1}/{r_{j_{i+1}^*}} + {1}/{f_{j_{i+1}^*}}		\right) d_{i+1}} > 	\frac{1}{\left(		{1}/{r_{j_{I}^*}} + {1}/{f_{j_{I}^*}}		\right) d_{i+1}}
	> 	\frac{(1-z_{j_{i}^*})^\alpha}	{	\left(		{1}/{r_{j_{i}^*} }+ {1}/{f_{j_{i}^*}}				\right) 				d_{i+1}			}, \forall i \in \cI$ with $\alpha =1$. 
	Therefore,  \eqref{selection} is used to select a new node $j_{i+1}^*$ to process task $i+1$. 
	Hence, a different neighboring node is selected for each task. 
\end{proof}
For instance, the condition in Lemma~3 can be satisfied if the value of $d_i$ is decreasing over time. 
In that case, every neighboring node can be used to process different tasks, thus satisfying constraint \eqref{const3}. 
As a last step, we derive the increment rate of the $\Delta P/\Delta D$ when a new task $i$ arrives in an online manner. 
\begin{lemma}\label{lemma3}
	When the dual problem's objective function increases by one, the primal problem's objective function increases by $ \frac{t_{\textrm{tot}}}{\left({1}/{r_j} + {1}/{f_j}\right) d_i} \left( 1 + \frac{1}{c-1} \right)  +  \Delta u_i$ for any given $c>1$. %
\end{lemma}\vspace{-3mm}
\begin{proof}
	See Appendix~\ref{appendix3}. 
\end{proof}
\noindent

Now, to derive a competitive radio for the proposed algorithm, we will adopt a primal-dual online analysis analogous to the one done in \cite{buchbinder2009design}. 
In Lemma~\ref{lemma1}, it is shown that the primal variable is updated while satisfying the constraints \eqref{const6} and \eqref{const7}. 
Then, we show that the dual constraints from \eqref{const1} to \eqref{const5} are satisfied under the derived condition in Lemma~\ref{lemma2a}. 
Finally, the increment rates of the primal and dual problems are, respectively, derived in Lemma~\ref{lemma3}. 
As a result, %
from Lemmas \ref{lemma1}, \ref{lemma2a}, and \ref{lemma3}, we obtain the following key result: %
\begin{theorem}\label{theorem2}
	\textcolor{black}{The upper bound of the competitive ratio} in Algorithm~\ref{algorithm} is 
	$\cO(1/\min_i \beta_{ij})$ where $\beta_{ij} \triangleq \left(\frac{1}{r_{j}} + \frac{1}{f_{j}}\right) \frac{d_{i}}{t_{\textrm{tot}}}$ if $d_i > \big( {\left({1}/{r_{j_{i}^*}} + {1}/{f_{j_{i}^*}}\right)^{-1}} - {\left({1}/{r_{j_{I}^*}} + {1}/{f_{j_{I}^*}}\right)^{-1}} \big) t_{\textrm{tot}}(c-1)$. 
\end{theorem}\vspace{-3mm}
\begin{proof}
	Lemma 1 first shows that the constraints of problem (P) are satisfied for all tasks that are assigned to the set of edge nodes. 
	At each iteration, Lemma~3 shows that the increment of $\Delta P/\Delta D$ is at most  \vspace{-3mm}
	\begin{equation}
			\frac{\Delta P}{\Delta D} \leq \frac{1}{\min_i \beta_{ij}} \left(1+ \frac{1}{  (1+\delta)^{\frac{1}{\delta}} -1 }  \right)   + \max_i  \Delta u_i,\label{thm2_1} \vspace{-3mm}
	\end{equation}
	where $\beta_{ij} = \left(\frac{1}{r_{j}} + \frac{1}{f_{j}}\right) \frac{d_{i}}{t_{\textrm{tot}}}$. 
	Also, \eqref{thm2_1} has an upper bound at $\delta = 1$. 
Since $D_{IP}=\sum_{\forall i,j} y_{ij}$,  the future tasks $i > D_{IP}$ cannot be allocated to any neighbor. 
In that case, Algorithm~\ref{algorithm} sets $\Delta u_i=1$. 
Then, all values of $u_{i'}, \forall i' \leq i$ increase by one, resulting in $\Delta D=0$ and $\Delta P=1$. 
Thus, we have $\gamma \leq \frac{\Delta P}{\Delta D} + (I - D_{IP})$. 
We observe that ${\Delta P}/{\Delta D}$ increases with the rate of $\cO({1}/{\min_i \beta_{ij}})$ as $\beta_{ij} \rightarrow 0$. At the same time, $I - D_{IP}$ can decrease with $ D_{IP}$ when the number of processed tasks increases. 
	Hence, the ratio $\gamma$ can be bounded by $\cO({1}/{\min_i \beta_{ij}})$. 
	\end{proof}
\noindent This result characterizes the online performance bound achieved by Algorithm~\ref{algorithm} in which \textcolor{black}{$\gamma$} can decrease as $\min_i \beta_{ij}$ approaches $1$. 
If $\min \left(\frac{1}{r_{j}} + \frac{1}{f_{j}}\right) \frac{d_{i}}{t_{\textrm{tot}}} \approx 1$, we have an environment in which all neighboring edge nodes have similar communication and computing performance, thus resulting in the smallest \textcolor{black}{$\gamma$} close to 1. In such a case, the online and offline performance gap is minimized. 
Also, Algorithm~\ref{algorithm} can be usefully converted into another simple algorithm that updates $\Delta u_i=0$ for all tasks $i\in\cI$ so that problem (P) has a value of $I$, by assuming  $\left(\frac{1}{r_{j}} + \frac{1}{f_{j}}\right) \frac{d_{i}}{t_{\textrm{tot}}}, \forall i, j, $ equals to $1$. 
This algorithm shows that the competitive ratio is inherently upper bounded by $P_{\textrm{LP}}/D_{\textrm{IP}}=I$ in the worst case. 

\begin{figure}
	\centering
	\includegraphics[width=0.45\textwidth]{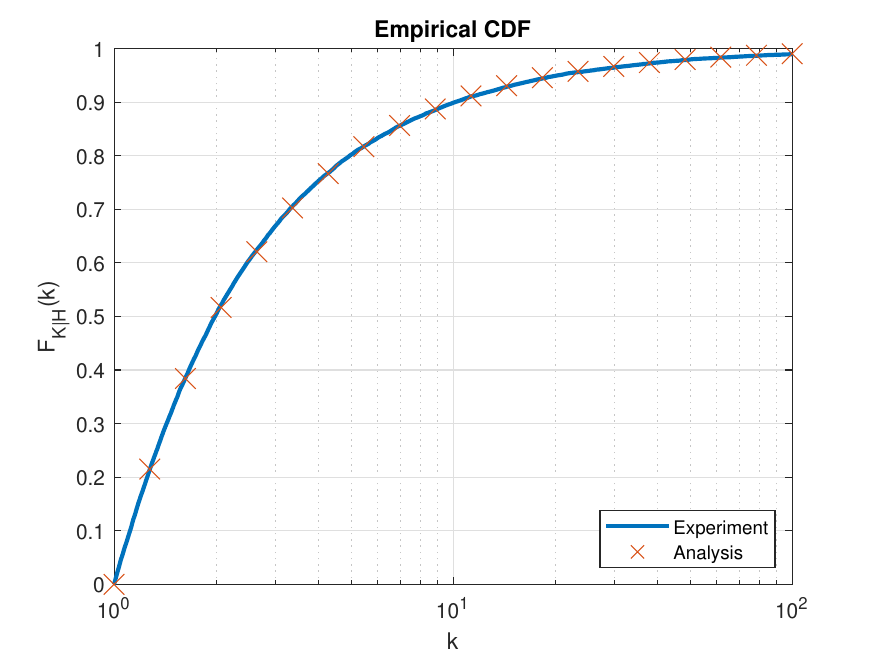}\vspace{-3mm}
	\caption{ \small Example of the cumulative probability distribution of $F_{K|H}(k)$.} %
	\label{fig_t1}\vspace{-7mm}
\end{figure}

	As shown in Theorem~1, it is essential to investigate how the value of $1/ \beta_{ij}$ is determined when measuring a realistic performance of the proposed ephemeral edge computing system. 
	We conduct a statistical analysis to derive the probability corresponding to different values of $1/ \beta_{ij}$. 
	To this end, it is assumed that the data rate and task size are randomly determined. 
	In particular, the size of data $d_i$ is generated by following a uniform distribution  random variable $D\sim U(0,D_{\textrm{max}})$ where $D_\textrm{max}$ is the maximum size of a task.
	We assume that the data rate is denoted by a random variable $R\triangleq \log_2(1+P)$ where $P$ is the received power in a fading channel modeled as an exponential distribution with parameter $\lambda$, i.e., $P\sim \textrm{exp}(\lambda)$. 
	This statistical model is a simplified version of our edge computing system model. 
This statistical modeling facilitates the observation of factors that affect the performance of the proposed algorithm. 
	Then, we derive the probability to have a certain value of $1/ \beta_{ij}$. %
	\begin{theorem}
		If $k \geq \frac{t_{\normalfont  \textrm{tot}} f}{D_{\normalfont \textrm{max}}}$, the probability that $1/ \beta_{ij} \leq k$ is $( F_K(k) - F_K(1) )/ (1-F_K(1))$ where 
		\begin{equation}
		F_{K}(k)  =  \frac{1}{D_{\textrm{max}}} \left[
		\int_{0}^{\frac{t_{\textrm{tot}} f}{k}} 
		\left(1-\exp\left( -\lambda \left(2^{\frac{1}{ 			\frac{t_{\textrm{tot}}}{kx} -\frac{1}{f}	}} -1\right) \right) \right) dx
		+ \left(D_{\textrm{max}} - \frac{t_{\textrm{tot}} f}{k} \right)
		\right].
		\end{equation}
	\end{theorem}
	\begin{proof}

		We define a random variable $K \triangleq \frac{t_{\textrm{tot}}}{\left(\frac{1}{\log_2 (1+P)}+\frac{1}{f}\right)D}$. 
		Therefore, if $k \geq \frac{t_{\textrm{tot}} f}{D_{\textrm{max}}}$, the cumulative density function of a random variable $K$ is shown as: \vspace{-3mm}
		\begin{eqnarray}
		F_{K} (k) &=& \textrm{Pr}\left(\frac{t_{\textrm{tot}}}{\left(\frac{1}{\log_2 (1+P)}+\frac{1}{f}\right)D}> k \right ) \\
		&=&\int_{0}^{D_{\textrm{max}}} \textrm{Pr}\left(\frac{t_{\textrm{tot}}}{\left(\frac{1}{\log_2 (1+P)}+\frac{1}{f}\right)x}> k | D=x\right) \textrm{Pr}(D=x) dx \\
		&=&  \frac{1}{D_{\textrm{max}}} \left[
		\int_{0}^{\frac{t_{\textrm{tot}} f}{k}} 
		\left(1-\exp\left( -\lambda \left(2^{\frac{1}{ 			\frac{t_{\textrm{tot}}}{kx} -\frac{1}{f}	}} -1\right) \right) \right) dx
		+ \left(D_{\textrm{max}} - \frac{t_{\textrm{tot}} f}{k} \right)		\right].
		\end{eqnarray}
		When $H$ is defined as the event in which $K \geq 1$, 		
		the cumulative density function of a random variable $K$ conditioned on $H$ is $
		F_{K|H} (k) = \frac{ \textrm{Pr}\left( K\leq k \cap K \geq 1 \right)}{ \textrm{Pr}(K\geq 1)} 
		=( F_K(k) - F_K(1) )/ (1-F_K(1)).$
	\end{proof}
	\noindent 	When the tasks are randomly generated and wireless performance dynamically changes, Fig.~\ref{fig_t1} shows an example of the cumulative probability distribution of $F_{K|H}(k)$ when $t_{\textrm{tot}} = 2$, $1/f = 0.5$, and $D_{\textrm{max}}=4$. 
	In Fig.~\ref{fig_t1}, if $k=2$, the probability that $k=1/\beta$ is less than $2$ is around $50$\%. 
	Therefore, the probability that $k$ becomes the empirical value of a competitive ratio in Theorem~2 is:
		$\textrm{Pr}(1/\min_i \beta_{ij} \leq k) = \textrm{Pr}(\max_i 1/\beta_{ij} \leq k) 
	= (F_{K|H}(k))^I. $ 
	Also, from Theorem~2, the derived probability does not change if the total time period is equal to the processing time of the maximum task size, i.e., $t_{\textrm{tot}}= D_{\textrm{max}}/f$. 
 Hence, if an ephemeral edge computing system is designed to use the maximum task size given by $t_{\textrm{tot}}f$, it is possible to expect the empirical value of the competitive ratio when the data rate and task size are randomly determined in a wireless environment. 

\vspace{-5mm}
\section{Simulation Results and Analysis}\label{sec:simulation}
For our simulations, we use a MATLAB simulator in which we consider that the source edge node initially forms a network with $J=10$ neighboring edge nodes uniformly distributed within a circular area of radius between $10~\text{m}$ and $100~\text{m}$. 
For instance, this can be seen as a generalized scenario in which an edge-enabled UAV (or vehicle) forms an edge network with $J$ neighboring nodes in a smart factory (or on a road environment). 
The task size follows a uniform distribution between 50 and 100~Mbits, and the number of tasks is $I=10$.
The power spectral density of the noise is -174~dBm/Hz, the carrier frequency is $2.1$~GHz, and $P_{t}=20$~dBm. 
The computational speed of each neighboring edge node is randomly determined from a uniform distribution between $1\times10^8$ and $5\times10^8$~bits/sec
\textcolor{black}{and we assume $r_j =  B \log_2 \left(1+\frac{g_j P_{t}}{\sigma^2}\right)$}.
The offline optimal solution is calculated by using a mixed-integer linear programming (MILP) solver with the assumption that the size $d_i$ of task $i$, $\forall i\in\cI,$ is completely known. 
\textcolor{black}{All simulations are statistically averaged over 5000 independent runs.}
%All statistical results are averaged over a large number of independent simulation runs. 

%
%
%
%
%
%
%
%
%
%
%
%
%
%
%
%
%
%
%

\begin{figure*}[t]\vspace{-2mm}
	\begin{multicols}{2}\vspace{-9mm}
		%\hspace{-4.5mm}
		\includegraphics[width=0.9\columnwidth]{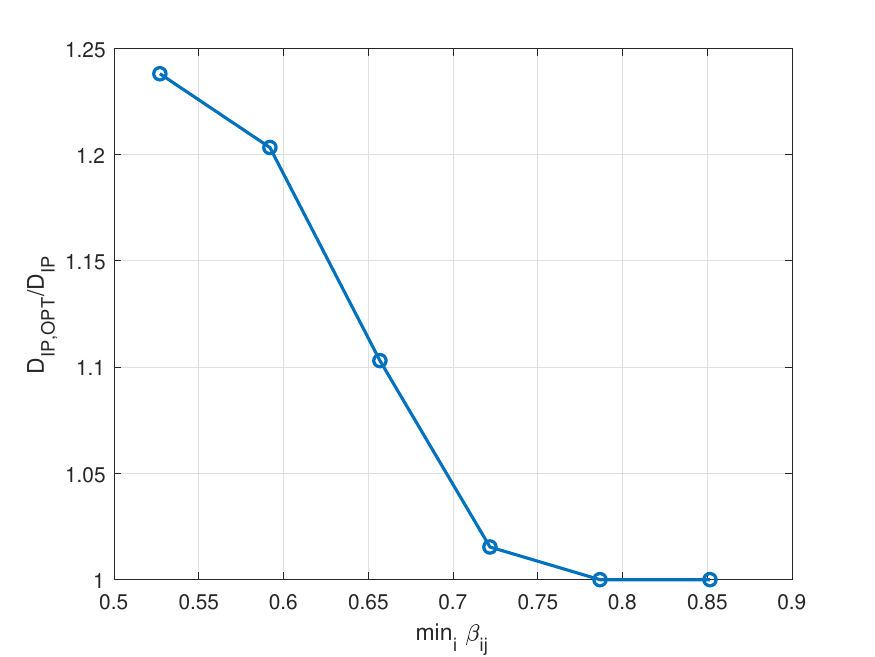}\vspace{-2.5mm}
		\par\caption{\small The empirical competitive ratio in \eqref{defcr} with respect to the different values of $\min_i \beta_{ij}$ when $\alpha=1$.}
\label{fig_beta}
		%\hspace{-4.5mm}
		\includegraphics[width=0.9\columnwidth]{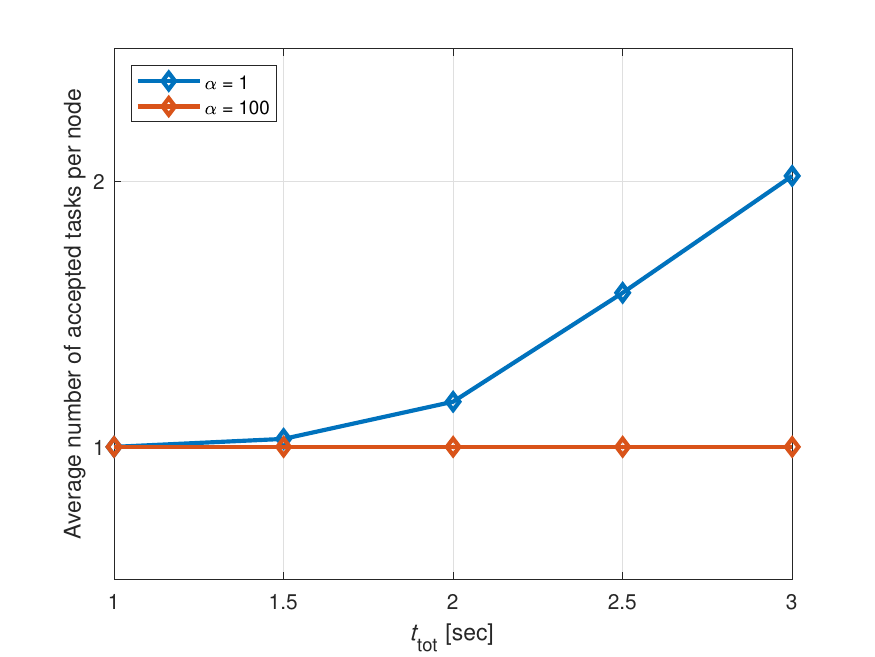}\vspace{-2.5mm}	
		\par\caption{\small Average number of computed tasks for the different total times when $\alpha=1$ and $100$.}
\label{fig1_a}
		\end{multicols}\vspace{-10mm}
\end{figure*}

Fig.~\ref{fig_beta} first shows the empirical ratio between the offline optimal and online solutions, $\textrm{D}_{\textrm{IP,OPT}}/\textrm{D}_{\textrm{IP}}$ for the different values of $\min_i \beta_{ij}$ when $t_{\textrm{tot}}=1$, $\alpha =1$, and $f_j \in[7\times10^7, 10\times10^7]$. 
The numerical results in Fig.~\ref{fig_beta} confirm that the ratio $\textrm{D}_{\textrm{IP,OPT}}/\textrm{D}_{\textrm{IP}}$ decreases as $\min_i \beta_{ij}$ increases as shown in Theorem~1. 
For example, the empirical competitive ratio can be reduced up to $19.2$\% if the smallest $\beta_{ij}$ increases from 0.58 to 0.85. 
Also, in Fig.~\ref{fig_beta}, the cases in which the ratio is one correspond to scenarios in which the proposed algorithm finds the optimal solution. 
For instance, when $\min_i \beta_{ij}$ is greater than $0.79$, Fig.~\ref{fig_beta} shows that the empirical ratio becomes one since $\textrm{D}_{\textrm{IP,OPT}}=\textrm{D}_{\textrm{IP}}$.

Fig.~\ref{fig1_a} shows the average number of accepted tasks per node, i.e., $\sum_i y_{ij}$, for two values of $\alpha= 1$ and $100$. 
In Fig.~\ref{fig1_a}, the number of accepted tasks per node needs to be one due to  constraint \eqref{const5}. 
When $0<z_j<1$, the selection rule in \eqref{selection} can decide to offload a new task to a neighboring node that already accepted a  task. 
In particular, Fig.~\ref{fig1_a} shows that the average number of accepted tasks per node increases with $t_{\textrm{tot}}$ for $\alpha=1$. 
This is due to the fact that the selection rule in \eqref{selection} is affected by two factors, i.e., $(1-z_j)^\alpha$ and $1/\left((1/r_j +1/f_j)\frac{d_i}{t_{\textrm{tot}}}\right)$ where $(1-z_j)^\alpha$ prevents the algorithm from choosing the same node multiple times. 
It is observable that $1/\left((1/r_j +1/f_j)\frac{d_i}{t_{\textrm{tot}}}\right)$ increases as $t_{\textrm{tot}}$ increases. 
Therefore, with a large $t_{\textrm{tot}}$, the selection rule in \eqref{selection} is determined by $1/\left((1/r_j +1/f_j)\frac{d_i}{t_{\textrm{tot}}}\right)$, rather than $(1-z_j)^\alpha$. 
For example, Fig.~\ref{fig1_a} shows the average number of accepted tasks can reach up to 2 when $t_{\textrm{tot}}$ increases from 1 to 3. 
Thus, to avoid offloading more than one task to the same neighboring node, a large $\alpha$ is used in Fig.~\ref{fig1_a}. 
If $\alpha$ is set to a large value, e.g., 100, Fig.~\ref{fig1_a} shows that the selection rule in \eqref{selection} only offloads the tasks to different nodes. 
This is due to the fact that $(1-z_j)^{\alpha}$ is close to zero for a large $\alpha$ when $0<z_j<1$. 
For instance, when $\alpha = 100$, the average number of accepted tasks is 1 for all $t_{\textrm{tot}}$. 
To evaluate Algorithm~\ref{algorithm} in a general task arrival, $\alpha = 100$ is used for the rest of our simulations. %

\begin{figure*}[t]\vspace{-2mm}
	\begin{multicols}{2}\vspace{-9mm}
		%\hspace{-4.5mm}
		\includegraphics[width=0.9\columnwidth]{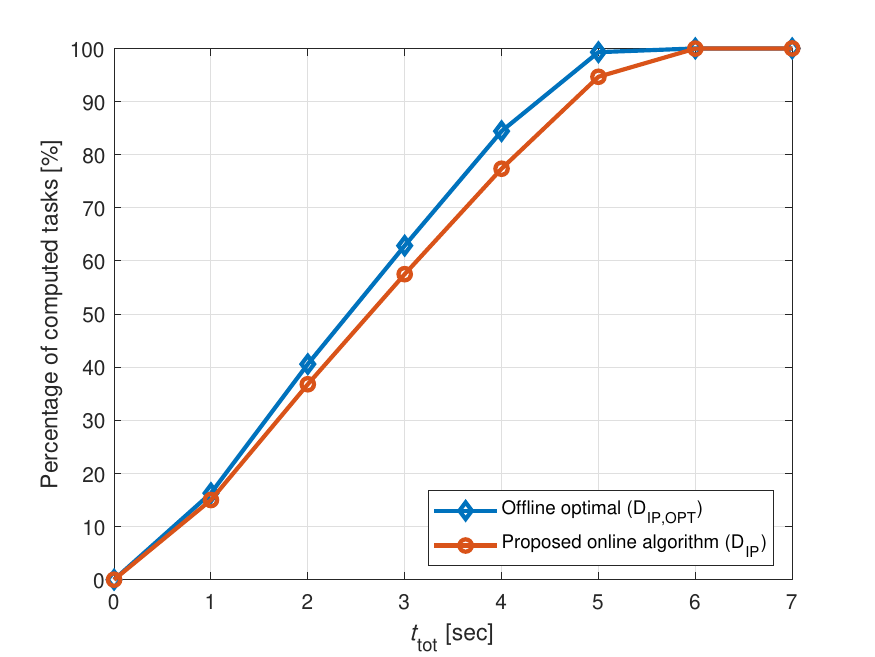}\vspace{-2.5mm}
		\par\caption{\small Comparison between the proposed algorithm's result and the offline optimal solution in terms of percentage of computed tasks for different $t_{\textrm{tot}}$.}
\label{fig1}
		%\hspace{-4.5mm}
		\includegraphics[width=0.9\columnwidth]{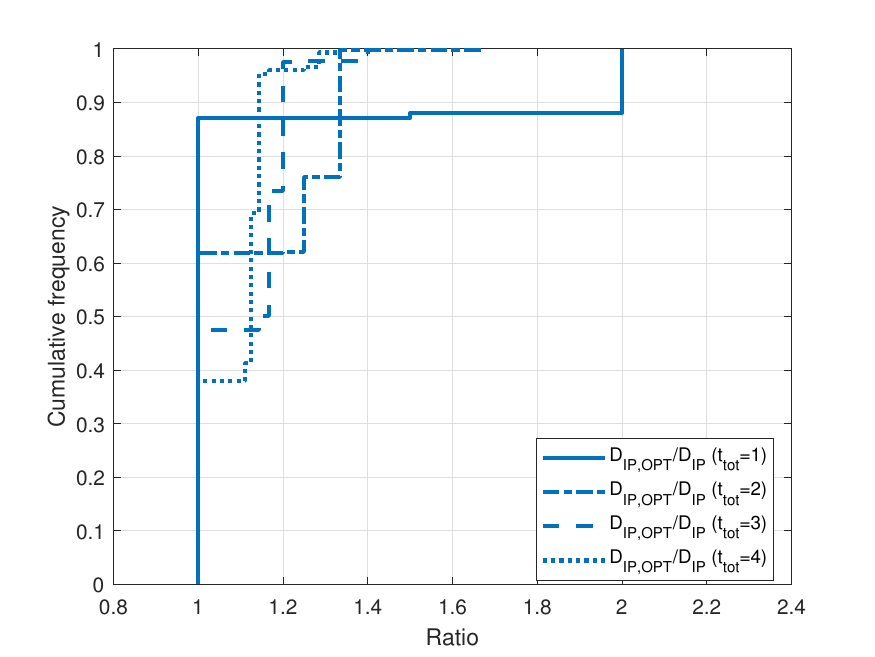}\vspace{-2.5mm}	
		\par\caption{\small The empirical competitive ratios $\textrm{D}_{\textrm{IP,OPT}}/\textrm{D}_{\textrm{IP}}$ when $t_{\rm{tot}}=1, 2,3,4$.}
\label{fig2_c}
		\end{multicols}\vspace{-10mm}
\end{figure*}

Fig.~\ref{fig1} shows the percentage of computed tasks for different values of $t_{\textrm{tot}}$ from 0 to 7~seconds when the total bandwidth is $10$~MHz. 
For comparison, we calculate the offline optimal solution of the dual integer problem, i.e., $\textrm{D}_{\textrm{IP,OPT}}$, by assuming that all task sizes, $d_i$, $\forall i$, are known in advance. %
The offline optimal $\textrm{D}_{\textrm{IP,OPT}}$ shows that the percentage of computed tasks increases with $t_{\textrm{tot}}$ that is a given parameter in problem (D). 
The design goal of our online algorithm is to achieve a performance that is similar to the offline optimal when the task size $d_i$ is revealed one by one. 
To this end, in Fig.~\ref{fig1}, we can observe that the optimal solution and the solution found by Algorithm~\ref{algorithm} are very close for all values of $t_{\textrm{tot}}$. 
This demonstrates the effectiveness of the proposed algorithm that can select properly neighboring edge nodes to offload tasks while maximizing the number of computed tasks. 
For instance, Fig.~\ref{fig1} shows that the maximum gap between the offline optimality and the online solution is only $7.1$\% when $t_{\textrm{tot}}=4$. 
Also, in Fig.~\ref{fig1}, as $t_{\textrm{tot}}$ increases, more tasks can be readily processed within a given time period, and, therefore, the percentage of computed tasks approaches to 100\%. 
In particular, when $t_{\textrm{tot}}=7$, Fig.~\ref{fig1} shows that all computational tasks are processed on the edge computing network in both online and offline cases, respectively.

\textcolor{black}{
Fig. \ref{fig2_c} shows the cumulative frequency of the empirical ratio, $\textrm{D}_{\textrm{IP,OPT}}/\textrm{D}_{\textrm{IP}}$, for both the offline optimal and online solutions when $t_{\rm{tot}}=1, 2,3,4$.
In Fig. \ref{fig2_c}, the ratio $\textrm{D}_{\textrm{IP,OPT}}/\textrm{D}_{\textrm{IP}}$ is shown to have a step-like shape since both $\textrm{D}_{\textrm{IP,OPT}}$ and $\textrm{D}_{\textrm{IP}}$ are integers, and there exists a limited number of possible values for $\textrm{D}_{\textrm{IP,OPT}}/\textrm{D}_{\textrm{IP}}$ for specific settings of  the simulations. 
In Fig. \ref{fig2_c}, the cases in which the ratio is one correspond to scenarios in which the proposed algorithm finds the optimal solution. 
For example, in Fig. \ref{fig2_c}, about $38-88$ \% iterations result in the slope of 1 where the optimal solution is achieved by running the proposed algorithm. 
By the definition of $\gamma$ in (\ref{defcr}), the number of computed tasks with the proposed algorithm is at least $\textrm{D}_{\textrm{IP,OPT}}/\gamma$. 
For instance, in Fig. \ref{fig2_c}, the largest empirical competitive ratio is shown to be $2$ which implies that the number of computed task is at least  $\textrm{D}_{\textrm{IP,OPT}}/2$ when the proposed algorithm is executed with the given simulation parameters.}

\begin{figure*}[t]\vspace{-2mm}
	\begin{multicols}{2}\vspace{-9mm}
		%\hspace{-4.5mm}
		\includegraphics[width=0.9\columnwidth]{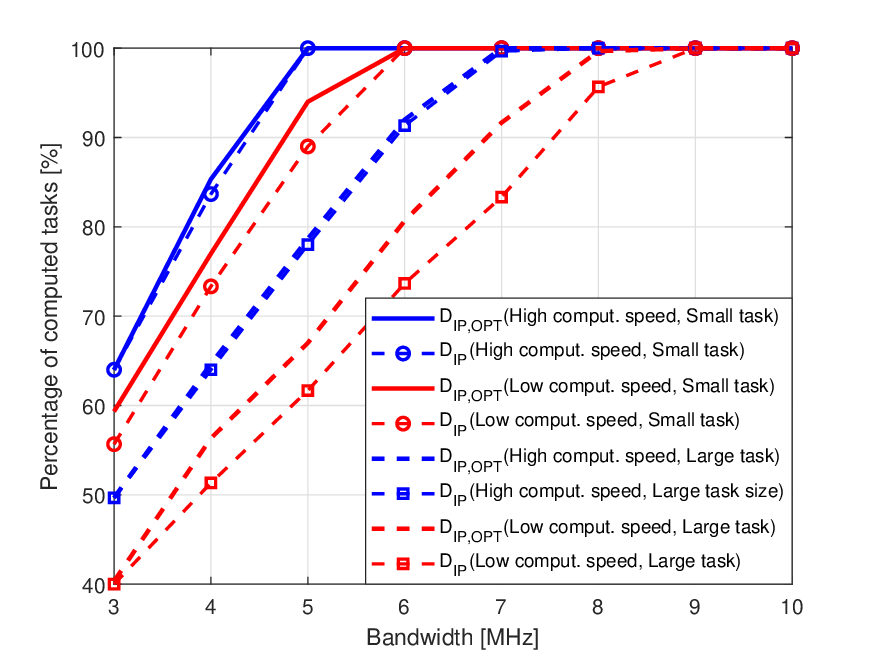}\vspace{-2.5mm}
		\par\caption{\small {Percentage of computed tasks for different computational speeds of neighboring edge nodes and different task sizes when  bandwidth is varying between $3$~and~$7$~MHz.}}
\label{fig3}
		%\hspace{-4.5mm}
		\includegraphics[width=0.9\columnwidth]{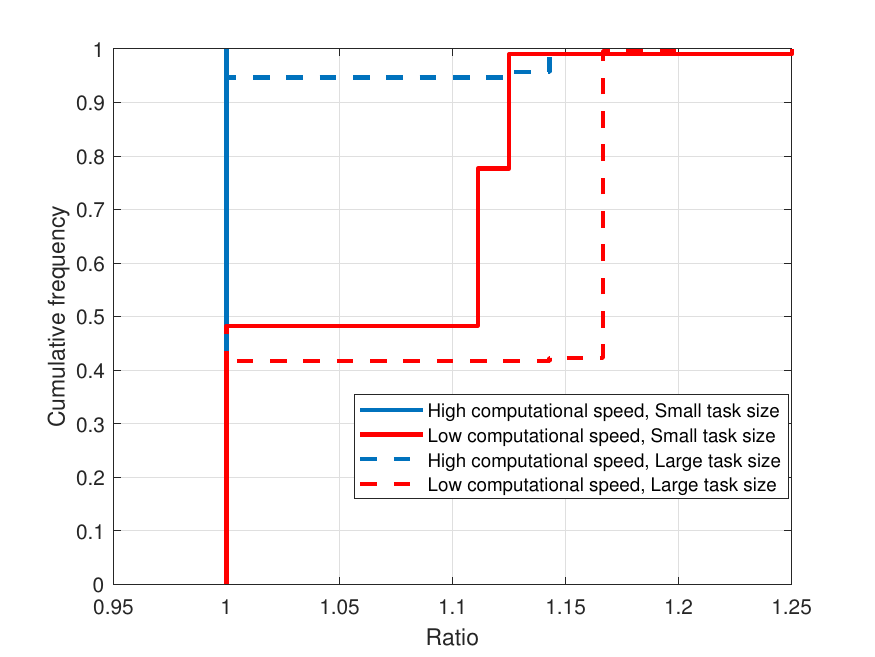}\vspace{-2.5mm}	
		\par\caption{\small The empirical competitive ratios $\textrm{D}_{\textrm{IP,OPT}}/\textrm{D}_{\textrm{IP}}$ for different computational speeds of neighboring edge nodes and different task sizes when  bandwidth is 5 MHz.}
\label{fig.cdfs}
		\end{multicols}\vspace{-10mm}
\end{figure*}

Fig.~\ref{fig3} shows the percentage of computed tasks for two different ranges of  computational speeds of the edge nodes and different task sizes when the bandwidth is changed from 3 to 10~MHz with $t_{\textrm{tot}}=7$ and distance randomly distributed in range from 10~m to 70~m. 
In Fig.~\ref{fig3},  neighboring edge nodes with low computational speeds are represented by $f_j \in[5\times10^7, 8\times10^7]$, whereas  edge nodes with high computational speeds are assumed to have $f_j\in[5\times10^8, 8\times10^8]$. 
Also, we consider two scenarios with small-size tasks $d_i\in[50\times10^6, 70\times10^6]$ and large-size tasks  $d_i\in[70\times10^6, 90\times10^6]$, respectively. 
From Fig.~\ref{fig3}, we can see that the number of computed tasks increases with more bandwidth. 
This is due to the fact that a higher bandwidth can increase the data rate and reduces tasks' transmission latency. Therefore, more tasks can be allocated to neighboring edge nodes. 
For instance, the number of computed tasks can increase about two-fold if the bandwidth changes from 3~MHz to 10~MHz in the case of  edge nodes with  low computational speeds and large-size tasks. 
Also, Fig.~\ref{fig3} shows that using edge nodes with high computational speeds increases the number of computed tasks. 
For example, the percentage of computed tasks increases from $88$\% to $99.5$\% by using edge nodes having high computational speeds when bandwidth is $5$~MHz and the task sizes are small. 
Moreover, Fig.~\ref{fig3} shows that more tasks can be computed as task sizes become smaller; for example, small-sized tasks result in $32.8$\% more computed tasks compared to that of large-sized tasks in the case of $4$~MHz in a high computational speed case.

\textcolor{black}{Fig.~\ref{fig.cdfs} shows the empirical competitive ratio $\textrm{D}_{\textrm{IP,OPT}}/\textrm{D}_{\textrm{IP}}$ for different computational speeds of neighboring edge nodes and different task sizes when bandwidth is 5 MHz.
We can observe that the proposed algorithm in both cases of edge nodes having high computational speeds
almost achieves the optimal performance that can be achieved by the offline optimal solution, i.e., $\textrm{D}_{\textrm{IP,OPT}}$.
However, as shown in Fig.~\ref{fig3}, since $\textrm{D}_{\textrm{IP,OPT}}$ in case of edge nodes having high computational speeds with large-sized tasks
is lower than that in case of edge nodes having low computational speeds with small-sized tasks,
the percentage of computed tasks in case of edge nodes having low computational speeds with small-sized tasks
is higher than that in case edge nodes having high computational speeds with large-sized tasks.}
\textcolor{black}{A higher computational capability can be achieved in a larger edge network than in a small one.
However, establishing a large network will increase the signaling overhead as the number of participating nodes increases. 
Hence, between a large and small edge network, there clearly exists a tradeoff between signaling overhead and computing capability.}

\begin{figure*}[t]\vspace{-2mm}
	\begin{multicols}{2}\vspace{-9mm}
		%\hspace{-4.5mm}
		\includegraphics[width=0.9\columnwidth]{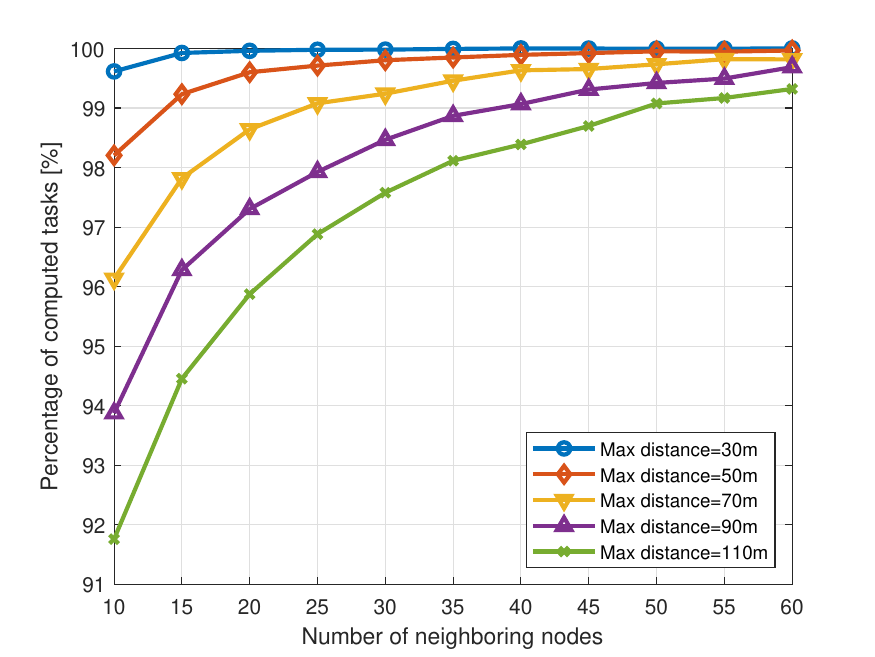}\vspace{-2.5mm}
		\par\caption{\small Percentage of computed tasks for different number of neighboring edge nodes and different maximum communication distances.}
\label{fig4}
		%\hspace{-4.5mm}
		\includegraphics[width=0.9\columnwidth]{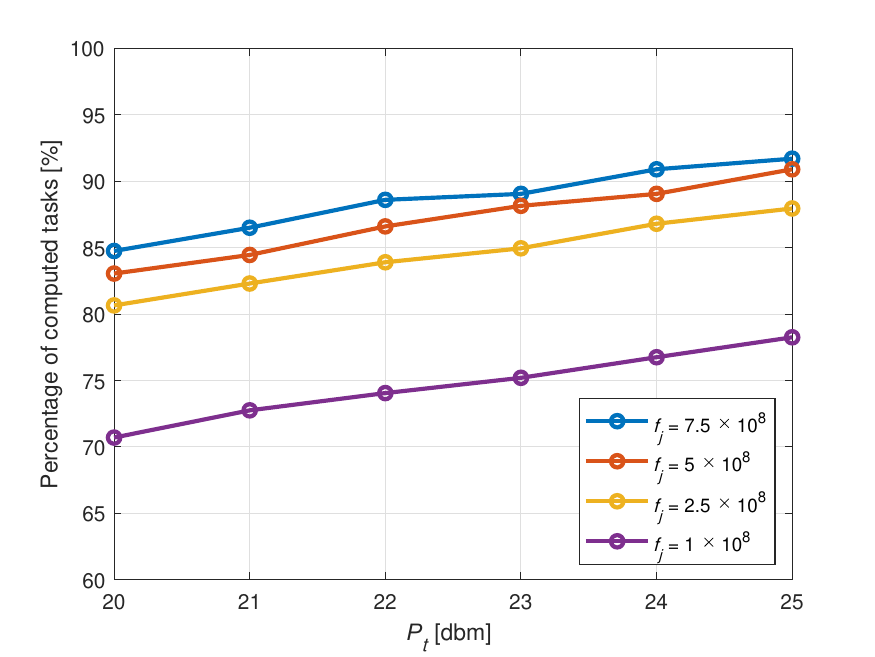}\vspace{-2.5mm}	
		\par\caption{\small Percentage of computed tasks for different transmit powers with respect to  different computing speeds of neighboring nodes.}
\label{fig5}
		\end{multicols}\vspace{-10mm}
\end{figure*}

In Fig.~\ref{fig4}, the percentage of computed tasks is shown for different numbers of neighboring edge nodes ranging from 10 to 60. 
The scenario in Fig.~\ref{fig4} assumes that neighboring edge nodes are randomly distributed within a maximum distance that is varied in range from 30~m to 110~m with $I=10$, $B=5$~MHz, and $t_{\textrm{tot}}=7$. 
Simulations assume that the small-size tasks are in the range of $d_i\in[40\times10^6, 70\times10^6]$. 
Also, the neighboring nodes use low computational speeds in the range of $f_j \in[5\times10^7, 8\times10^7]$. %
In Fig.~\ref{fig4}, it is clear that the number of computed tasks increases with the number of neighboring edge nodes. 
As the  set of neighboring edge nodes becomes larger, the source edge node has a higher probability to allocate its tasks to the neighboring edge nodes having a high data rate and computational speed. 
For instance, the number of computed tasks can increase by about $8.2$\% if the number of edge nodes increases from 10 to 60 when the maximum distance is $110$~m.  
Fig.~\ref{fig4} also shows that the number of computed tasks increases if the maximum communication distance between edge nodes is reduced. 
For example, the percentage of computed tasks increases from $91.8$\% to $99.6$\% by reducing the maximum distance between neighboring edge nodes and the source edge node.

In Fig.~\ref{fig5}, the percentage of computed tasks is shown for different transmit powers from 20~dBm to 25~dBm when the neighboring nodes use identical computing speed that varies from $10^8$ to $7.5\times 10^8$. 
Fig.~\ref{fig5} shows that  the number of computed tasks increases with the transmit power of the source edge node. 
This is due to the fact that the increased data rate reduces the wireless transmission latency, and, therefore, more tasks can be processed within a limited time period. 
For example, the percentage of computed tasks increases by up to $10.7$\% if the transmit power changes from 20~dBm to 25~dBm with $f_j = 10^8$. %
Also,  Fig~\ref{fig5} shows that increasing a computing speed is beneficial to process notably more tasks. 
For instance,  if the computing speed of edge nodes increases from $10^8$ to $7.5\times 10^8$, the edge computing network can process up to about $20\%$ more tasks. %
Thus, Fig~\ref{fig5} shows that reducing the computing latency by using a high computing speed is needed while reducing the transmission latency with a high power.

\vspace{-5mm}
\section{Conclusion}\label{sec:conclusion}

In this paper, we have proposed a new concept of ephemeral edge computing in which the total time period dedicated to edge computing is limited. 
This concept of ephemeral edge computing is applicable to a wide range of scenarios including  Industry 4.0 smart factory, intelligent transportation systems, and smart homes. 
By modeling a generalized scenario of ephemeral edge computing, we have proposed a novel framework to maximize the number of successful computations over an edge computing network within a limited time period. 
This framework allows a source edge node to offload tasks from sensors and distributed tasks to neighboring edge nodes in order to compute the tasks before the source edge node discontinues its current edge computing network. 
When the exact information on the offloaded tasks is unknown to the source edge node, it is challenging  to optimize the decision of which neighboring edge node has to compute each task. 
Therefore, we have formulated an online optimization problem that jointly optimizes the communication and computation latency is formulated and introduced an online greedy algorithm to solve the problem. 
Then, by using the structure of the primal-dual problem formulation, we have derived a feasible competitive ratio  as a function of the task sizes and the data rates of the edge nodes.
Simulation results have shown that the  empirical competitive ratio defined as the ratio between the number of computed tasks achieved by the proposed online algorithm and offline optimal case is at most 2 in a given simulation setting. 
Thus, the simulation results confirm that the proposed online algorithm can efficiently allocate tasks to neighboring edge nodes under uncertainty.
\textcolor{black}{Our future work will include extending our results to additional practical scenarios in which multiple tasks can be allocated to edge nodes under the consideration of the lifetime of an ephemeral edge computing network.}

\vspace{-5mm}

\appendices

\section{Derivation of~(\ref{selection})}\label{appendixA}\vspace{-1mm}

\textcolor{black}{In (\ref{problemD}), problem (D) is formulated to minimize the sum delay.
In this regard, it is beneficial to offload task $i$ from the source node to the neighbor with a high data rate and computing speed
to minimize the sum delay.
Therefore, the decision rule to select edge node $j^*$ needs to be designed to select an edge node with the shortest communication and computing latency to process the task $i$, i.e., $\left(\frac{1}{r_j} + \frac{1}{f_j}\right) d_i$. 
In (\ref{problemP}), problem (P) is formulated to minimize the cost objective function $\sum_{i=1}^I t_{\textrm{tot}} x_i {+} \sum_{j=1}^J z_j {+} u_1$.
Then, the decision rule to select edge node $j^*$ need to be designed to select an edge node with the smallest $z_j$ which is equivalent to
select an edge node with the largest $\left(1-z_j\right)^\alpha$.
Thus, the decision rule to select edge node $j^*$ is obtained by~(\ref{selection}).}

\vspace{-3mm}
\section{Proof of Lemma~\ref{lemma1}}\label{appendix1}\vspace{-1mm}

	We will show that the first constraint is always satisfied for all $i$ when using the updating rule. 
	When allocating task $i'$, $x_i = 0, \forall i\geq i'$ and $u_i = 0, \forall i$ due to the initialization. 
	From the constraints in \eqref{const6} and \eqref{const7}, we have that 
	$\left( {1}/{r_j} + {1}/{f_j} \right) d_i x_i + \left({d_i}/{r_j} \right)  \sum_{i'=i+1}^I x_{i'} + z_j + u_i - u_{i+1} 
	= \left( {1}/{r_j} + {1}/{f_j} \right) d_i (1-z_j) \frac{1}{\left({1}/{r_j} + {1}/{f_j}\right) d_i}  + z_j = 1.$
	Then, we consider the constraints regarding other edge nodes $j \in \cJ\setminus \{j^*\}$ for a given task $\forall i \in \cI$. 
	When $u_i$ is updated, $ u_i - u_{i+1}$ is equal to $\Delta u_i$. 
	Therefore, we can show that edge node $\forall j \in \cJ$ satisfy the constraint \eqref{const4} as follows:
	\begin{eqnarray}
	%&&\left( \frac{1}{r_j} + \frac{1}{f_j} \right) d_i x_i + \left(\frac{d_i}{r_j} \right)  \sum\nolimits_{i'=i+1}^I x_{i'} + z_j + u_i - u_{i+1} \\
	&&\left( \frac{1}{r_j} + \frac{1}{f_j} \right) d_i	 (1-z_{j^*})^\alpha \frac{1}{\left({1}/{r_{j^*}} + {1}/{f_{j^*}}\right) d_i} + z_j + \Delta u_i \nonumber\\
	&&=	 	 \frac{ \left( {1}/{r_j} + {1}/{f_j} \right)}{\left({1}/{r_{j^*}} + {1}/{f_{j^*}}\right) } (1-z_{j^*})^\alpha  + z_j +
	\max_{j' \in \cJ  } \left( 1-  \left(	\frac{\left( {1}/{r_{j'}} + {1}/{f_{j'}} \right) }{\left({1}/{r_{j^*}} + {1}/{f_{j^*}}\right) } (1-z_{j^*})^\alpha  + z_{j'}  \right), 0 \right)\nonumber \\
	&&\geq \frac{ \left( {1}/{r_j} + {1}/{f_j} \right)}{\left({1}/{r_{j^*}} + {1}/{f_{j^*}}\right) } (1-z_{j^*})^\alpha  + z_j +
	\left( 1-  \left(	\frac{\left( {1}/{r_{j}} + {1}/{f_{j}} \right) }{\left({1}/{r_{j^*}} + {1}/{f_{j^*}}\right) } (1-z_{j^*})^\alpha  + z_{j}  \right) \right) = 1. \vspace{-2mm}
	\end{eqnarray}
	Hence, the primal constraints \eqref{const6} and \eqref{const7} are satisfied.

\vspace{-3mm}
\section{Proof of Lemma~\ref{lemma2}}\label{appendix2}\vspace{-1mm}

	For a given $j$, the upper bound of $\sum_{\forall i} y_{ij}$ in \eqref{const3} is derived by using the fact that the proposed algorithm does not update $z_j$ if $\sum_i y_{ij} \geq 1$. 
	In particular, when the task is indexed by $i'$, suppose that the task allocation is not possible for the first time, i.e.,  $y_{ij} = 0$, $\forall i>i'$. 
	Before the last task $i'$ arrives, the value of $\sum_{\forall i} y_{ij}$ is still less than the total budget of edge node $j$. 
	However, after allocating task $i'$ to edge node $j$, $\sum_{\forall i} y_{ij}$ can be greater then one. 
	The violation of the constraint \eqref{const3} makes the value of $z_j$ be greater than 1. 
	Therefore, for any $c>1$, the inequality $z_{j} \geq \frac{1}{c-1} \left(  c^{\sum_{i=1}^{i'} y_{ij} } - 1\right)$ is used to derive the upper bound of $\sum_{\forall i} y_{ij}$. 
	From this relationship, if  $\sum_{i=1}^{i'} y_{ij} \geq 1$, we have $\frac{1}{c-1}\left(  c^{\sum_{i=1}^{i'} y_{ij} } - 1\right) \geq 1$, then $z_j \geq 1$. 
	
	When we define $\beta_{i'j} = \left(\frac{1}{r_{j}} + \frac{1}{f_{j}}\right) \frac{d_{i'}}{t_{\textrm{tot}}}$, the update rule of $z_j$ in \eqref{updatez} is used as following:
\vspace{-2mm}
	\begin{eqnarray}
	z_{j} &=& z_{j} (1+\beta_{i'j}) + \beta_{i'j} \frac{1}{c-1} \label{c3_1}\\
	& \geq & \frac{1}{c-1} \left( c^{ \sum_{i=1}^{i' -1 } y_{ij} } - 1 \right) (1 + \beta_{i'j}) + \beta_{i'j} \frac{1}{c-1} \label{c3_2}\\
	& = & \frac{1}{c-1} \left( c^{ \sum_{i=1}^{i' -1 } y_{ij} } (1 + \beta_{i'j}) - 1 \right) \label{c3_3}\\
	& \stackrel{(a)}{\geq} & \frac{1}{c-1} \left( c^{ \sum_{i=1}^{i' -1 } y_{ij} + \beta_{i'j}   }   - 1 \right), \label{c3_4}			
	\end{eqnarray}
	where $c \triangleq \left(1+\delta \right)^{		1/\delta}$ for a constant $\delta \geq \beta_{i'j} $. %
	From the definition of $c$, $(a)$ holds due to the relationship $1+\beta_{i'j} \geq \left(\left(1+\delta \right)^{		1/\delta}\right)^{ \beta_{i'j} } =\left( \left(1+\delta \right)^{		1/\delta}\right)^{ \beta_{i'j} }$ when $0\leq \beta_{i'j} \leq \delta \leq 1$. 
	Also, the definition of $z_j$ in \eqref{updatez} has an upper bound $z_j \leq \bar{z} \triangleq (1+\delta) + \frac{\delta}{c-1}$, and, therefore, we can rewrite  \eqref{c3_4} as following: 
	$\sum\nolimits_{i=1}^{i'-1} y_{ij} \leq \log_c\left( \bar{z}(c-1) +1\right) - \beta_{i'j}$.
%	\begin{equation}
%	\sum\nolimits_{i=1}^{i'-1} y_{ij} \leq \log_c\left( \bar{z}(c-1) +1\right) - \beta_{i'j}.\label{c3_5}
%	%
%	\end{equation}
Thus, an upper bound of $\sum\nolimits_{i=1}^{i'} y_{ij}$ is derived as: \vspace{-3mm}
	\begin{eqnarray}
	\sum\nolimits_{i=1}^{i'} y_{ij} \stackrel{(a)}{\leq}& \log_c (\bar{z}(c-1)+1) -\beta_{i'j}+1 
	\leq 1+\log_c \frac{(1+\delta)c}{c^{\beta_{i'j}}} \label{c3_7}
	\end{eqnarray}
	where	$(a)$ hold since $\sum_{i=1}^{i'} y_{ij} = \sum_{i=1}^{i' -1 } y_{ij}  +1$ if task $i'$ is allocated. 
	Then, if $\delta =\beta_{i'j}=0$, we can have a lower bound $1+\log_c \frac{(1+\delta)c}{c^{\beta_{i'j}}} = 2$. 

\vspace{-3mm}
\section{Proof of Lemma~\ref{lemma3}}\label{appendix3}\vspace{-1mm}

	By using  the definition of $z_j$ and $x_i$, we derive the change of the objective function of  problem (P), denoted by $\Delta P$. 
	When a task $i$ is allocated to an edge node $j$, $z_j$ and $x_i$ are updated, and, therefore, the objective function of problem (P) increases. 
	In particular, $\Delta P$ increase with $\Delta z_j$ since we want to observe the increment of $z_j$ at current interation while the value of $z_j$ can be updated multiple time. 
	Also, $\Delta P$ increase with $x_i$  since $x_i$ is initially given by 0 and updated only once. Thus, we have $\Delta P = \Delta z_{j} +  t_{\textrm{tot}} x_i + \Delta u_i$ and \vspace{-3mm}
	\begin{eqnarray}
	\Delta P &=& \left(\frac{1}{r_j} + \frac{1}{f_j}\right) \frac{d_i}{t_{\textrm{tot}}} \left( z_j + \frac{1}{c-1} \right) + 
	t_{\textrm{tot}} (1-z_j)^\alpha  \frac{1}{\left({1}/{r_j} + {1}/{f_j}\right) d_i}    + \Delta  u_i\nonumber\\			
	&\stackrel{(a)}{\leq} & \frac{t_{\textrm{tot}}}{\left({1}/{r_j} + {1}/{f_j}\right) d_i}  \left( z_j + \frac{1}{c-1} \right) + 
	(1-z_j)  \frac{t_{\textrm{tot}} }{\left({1}/{r_j} + {1}/{f_j}\right) d_i}  + \Delta  u_i\nonumber\\			
	& \stackrel{}{=} &   \frac{t_{\textrm{tot}}}{\left({1}/{r_j} + {1}/{f_j}\right) d_i} \left(1 + \frac{1}{c-1} \right)  + \Delta  u_i,
	\end{eqnarray}
	where $(a)$  holds due to $\left({1}/{r_j} + {1}/{f_j}\right) \frac{d_i}{t_{\textrm{tot}}} \leq 1$ with $\alpha = 1$. 	
	Next, the objective function of problem (D) is increases by one, and it is denoted by $\Delta D=1$. 
	This is due to the fact that $y_{ij}$ is initially set to zero, and we update $y_{ij}=1$ when task $i$ is assigned to edge node $j$. 
	Hence,  we have $\frac{\Delta P}{\Delta D} \leq \frac{t_{\textrm{tot}}}{\left({1}/{r_j} + {1}/{f_j}\right) d_i} \left(1+ \frac{1}{c-1}\right)  + u_i$.

\vspace{-5mm}
\bibliographystyle{IEEEtran}
\bibliography{UAV2020}

\end{document}